%% file: main.tex
\documentclass[sigconf, nonacm]{acmart}
\graphicspath{{img/}}
\usepackage[a-1b]{pdfx}
\usepackage{hyperref}
\usepackage{array}
\usepackage{amsmath}
\usepackage{algorithmic}
\usepackage[ruled,vlined]{algorithm2e}
\usepackage{pgfplots,pgfplotstable}  
\usepackage{booktabs}
\usepackage{balance}
\usepackage{multirow}
\pgfplotsset{compat=newest}

\newcommand{\stitle}[1]{\vspace{0.5ex}\noindent{\bf #1}}

\newcommand{\aj}[1]{\textcolor{blue}{[Alekh: #1]}}

\newcommand\vldbdoi{10.14778/3476249.3476298}
\newcommand\vldbpages{2505 - 2518}
\newcommand\vldbvolume{14}
\newcommand\vldbissue{11}
\newcommand\vldbyear{2021}
\newcommand\vldbauthors{\authors}
\newcommand\vldbtitle{\shorttitle} 

\newcommand\vldbpagestyle{empty} 
\newcommand{\eat}[1]{}
\newcommand{\msft}{Microsoft\xspace}
\newcommand{\cosmos}{Cosmos\xspace}
\newcommand{\scope}{SCOPE\xspace}

\newcounter{enum}
\newenvironment{packed_enum}{
	\begin{list}{\textbf{(\arabic{enum})}}{
			\setlength{\itemsep}{0pt}
			\setlength{\parskip}{0pt}
			\setlength{\labelwidth}{-4 pt}
			\setlength{\leftmargin}{0 pt}
			\setlength{\itemindent}{0pt}
			\usecounter{enum}}
	}{\end{list}}

\newcommand{\rev}[1]{{\color{black}#1}}

\begin{document}
\title{Phoebe: A Learning-based Checkpoint Optimizer}

\author{Yiwen Zhu}
\affiliation{%
  \institution{Microsoft}
  \streetaddress{1020 Enterprise Way}
  \postcode{43017-6221}
}
\email{yiwzh@microsoft.com}

\author{Matteo Interlandi}
\affiliation{%
	\institution{Microsoft}
	\streetaddress{One Microsoft Way}
}
\email{mainterl@microsoft.com}

\author{Abhishek Roy}
\affiliation{%
	\institution{Microsoft}
	\streetaddress{One Microsoft Way}
}
\email{abro@microsoft.com}

\author{Krishnadhan Das}
\affiliation{%
	\institution{Microsoft}
}
\email{krisdas@microsoft.com}

\author{Hiren Patel}
\affiliation{%
	\institution{Microsoft}
}
\email{hirenp@microsoft.com}

\author{Malay Bag}
\affiliation{%
	\institution{Facebook}
}
\email{malayb@gmail.com}

\author{Hitesh Sharma}
\affiliation{%
	\institution{Google}
}
\email{hitesh@outlook.com}

\author{Alekh Jindal}
\affiliation{%
	\institution{Microsoft}
}
\email{aljindal@microsoft.com}

\input{abstract}

\maketitle

\pagestyle{\vldbpagestyle}
\begingroup\small\noindent\raggedright\textbf{PVLDB Reference Format:}\\
\vldbauthors. \vldbtitle. PVLDB, \vldbvolume(\vldbissue): \vldbpages, \vldbyear.\\
\href{https://doi.org/\vldbdoi}{doi:\vldbdoi}
\endgroup
\begingroup
\renewcommand\thefootnote{}\footnote{\noindent
This work is licensed under the Creative Commons BY-NC-ND 4.0 International License. Visit \url{https://creativecommons.org/licenses/by-nc-nd/4.0/} to view a copy of this license. For any use beyond those covered by this license, obtain permission by emailing \href{mailto:info@vldb.org}{info@vldb.org}. Copyright is held by the owner/author(s). Publication rights licensed to the VLDB Endowment. \\
\raggedright Proceedings of the VLDB Endowment, Vol. \vldbvolume, No. \vldbissue\ %
ISSN 2150-8097. \\
\href{https://doi.org/\vldbdoi}{doi:\vldbdoi} \\
}\addtocounter{footnote}{-1}\endgroup

\input{Introduction2}

\input{Motivation}
\input{Background}

\input{Overview}

\input{CostModels}

\input{TTL}

\input{Optimizer}

\input{Results}

\input{Related}

\input{Conclusion}

\bibliographystyle{ACM-Reference-Format}
\bibliography{sample}

\end{document}

%% file: abstract.tex
\begin{abstract}

Easy-to-use programming interfaces paired with cloud-scale processing engines have enabled big data system users to author arbitrarily complex analytical jobs over massive volumes of data.
However, as the complexity and scale of analytical jobs increase, they encounter a number of unforeseen problems, hotspots with large intermediate data on temporary storage, longer job recovery time after failures, and worse query optimizer estimates being examples of issues that we are facing at Microsoft.

To address these issues, we propose Phoebe, an efficient learning-based checkpoint optimizer. 
Given a set of constraints and an objective function at compile-time, Phoebe is able to determine the decomposition of job plans, and the optimal set of checkpoints to preserve their outputs to durable global storage. Phoebe consists of three machine learning predictors and one optimization module.
For each stage of a job, Phoebe makes accurate predictions for: (1) the execution time, (2) the output size, and (3) the start/end time taking into account the inter-stage dependencies. %
Using these predictions, we formulate checkpoint optimization as an integer programming problem and propose a scalable heuristic algorithm that meets the latency requirement of the production environment.

We demonstrate the effectiveness of Phoebe in production workloads, and show
that we can free the temporary storage on hotspots by more than $70\%$ and 
restart failed jobs $68\%$ faster on average with minimum performance impact.
Phoebe also illustrates that adding multiple sets of checkpoints is not cost-efficient, which dramatically reduces the complexity of the optimization. 

\end{abstract}

%% file: Introduction2.tex
\section{Introduction}
\label{sec:intro}

\begin{figure}[!t]
	\centering
	\vspace{0.4cm}
	\includegraphics[width=0.75\columnwidth]{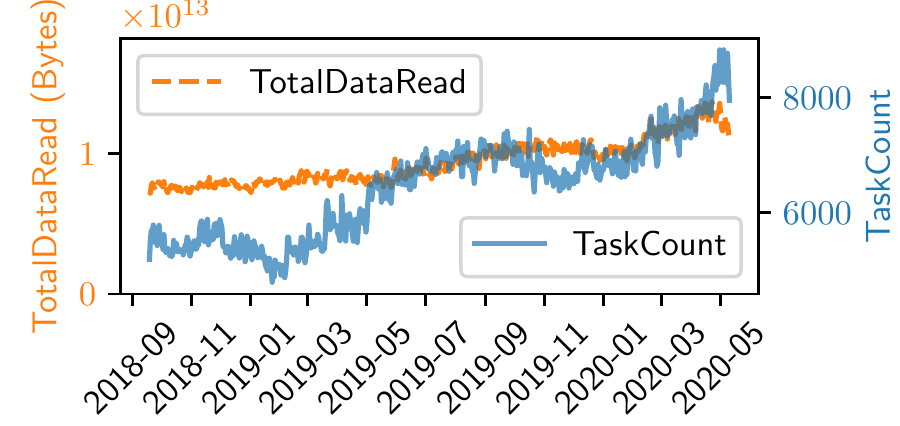}
	\vspace{-0.3cm}
	\caption{\cosmos job size}\label{fig:trend}
	\vspace{-0.5cm}
\end{figure}

Big data platforms have democratized scalable data processing over the last decade, giving developers the freedom of writing complex programs (also referred to as \emph{jobs}) without worrying about scaling them~\cite{zaharia2010spark, thusoo2009hive,tigani2014google,chaiken2008scope,diaz2018azure,ramakrishnan2017azure,aws-athena, sarkar2018learning}. 
However, this flexibility has also led developers into building very large analytical programs that can put the underlying platform under stress.
For instance, the scale and flexibility of \textit{\cosmos}
~\cite{curino2019hydra,shao2019griffon,chaiken2008scope}, 
a big data analytics platform at \msft,
empower developers to author large jobs composed of pipelines of SQL-like query statements which are compiled into query plans composed of up to thousands of \emph{stages} (executable units composed by one or more operators) running over hundreds of thousands of processing units scheduled by YARN~\cite{vavilapalli2013apache}.
Figure~\ref{fig:trend} shows how the \cosmos workloads in one of the clusters have evolved over the past two years: we see that the 
total number of \emph{tasks} per job (each corresponding to one process executed in one container) \rev{has} grown by $34\%$ (shown in blue), while the
volume of input data has grown by $80\%$ (shown in orange).
Large analytical jobs lead to several operational problems.

\stitle{1. Large jobs result in machine hotspots that run out of local storage space due to temporary data.} Big data systems typically persist intermediate outputs on local SSDs until the end of the query. However, large query plans end up consuming a substantial amount of temporary storage. 
Figure~\ref{fig:ssd} (left) shows, for one \cosmos cluster, the cumulative distribution of available local SSD storage that is used for storing temporary data.
We can see that for different Stock Keeping Units (SKUs), $15-50\%$ of the machines run out of local storage on SSDs. %
This results in not only expensive spilling to HDDs and hence processing slowdown, but also an increase in incidents reporting job failures due to SSD outages. %
Currently, to avoid the SSD shortage, we need to either 
cap the number of containers running on each machine (thus wasting expensive CPU and memory resources)~\footnote{Alternate solution could be to make YARN scheduler aware of the SSD utilization. However, additional parameters are not only harder to tune cluster-wide but also increase the scheduling overhead, which could
translate into high costs at scale.}, or 
scale the CPU and memory together with the temporary storage in the newer SKUs.

\stitle{2. Large jobs are prone to longer re-starting time in case of failures.} 
Figure~\ref{fig:failure} (right) shows the failure rates of jobs with increasing runtimes. {\color{black}We observe that even though a majority of the jobs finish within $20$ minutes, the failure rates in larger jobs could range as high as $5\%$.} Network/communication failures, changes in cluster conditions, transient system behavior, and user errors are some of the common reasons for job failures in Cosmos.
Even though Cosmos already provides a lineage-based mechanism~\cite{rdd} for coping with task failures, improving resiliency to job failures and reducing recovery time are very important, especially as jobs and workloads scale.

\stitle{3. Large jobs end up having worse query optimizer estimates.} Errors in cardinality propagate exponentially~\cite{moerkotte2009preventing,neumann2013taking}, and hence complex jobs are more likely to produce poor query plans. A recent trend~\cite{rios,adptive-query-spark} suggests to re-optimize plans adaptively during job execution, but collecting statistics on-the-fly and on distributed intermediate results is highly not trivial and requires \rev{a} major overhaul of the runtime system.

\rev{
\stitle{Checkpointing and its challenges.} 
The above problems can be solved by ``decomposing'' large jobs into smaller ones that are separated by persistent (e.g., with 3-way replication) checkpoints on durable global storage.
\eat{Our observation is that the above problems can all be solved with a simple, yet effective technique: \textit{checkpointing intermediate state to a durable global storage}. In fact, by ``decomposing'' large jobs into smaller ones that are separated by persistent (e.g., with 3-way replication) checkpoints, }
This will allow, for instance, to (1) free up intermediate data on hotspots even before the job completes; (2) fast restart failed jobs from the previous state; (3) collect statistics on the checkpoints and re-optimize large jobs into smaller ones with better estimates; and (4) reduce intermediate data in local storage to avoid wasting resources in newer SKUs. 
Prior checkpointing approaches include gathering statistics at execution time to dynamically select when to checkpoint~\cite{daly2006higher,sharma2016flint,yan2016tr}, however, they require additional dynamic components that are not easy to implement reliably in large production systems such as Cosmos.
Alternatively, compile-time approaches use estimates to propose optimal checkpoints during query optimization~\cite{upadhyaya2011latency,yang2010osprey,chen2013selective}.
However, this requires accurate cost estimates, which is challenging since query optimizer estimates are often off by several orders of magnitudes~\cite{dutt2019selectivity,MCSN,siddiqui2020cost}, and even learned approaches are good for only relative plan comparisons~\cite{marcus2019plan,marcus2019neo,ortiz2018learning} while still being significantly off in absolute values.
Furthermore, all previous checkpointing approaches considered relatively small tree-shaped plans, whereas modern big data systems like SCOPE easily have complex Direct Acyclic Graphs (DAGs) with thousands of operators~\cite{microlearner}.
Not to mention, we need a generalized framework that can make checkpointing decisions on these DAGs for different scenarios with different objectives and constraints.
}

\begin{figure}[t]
	\centering
	\begin{minipage}[b]{0.5\columnwidth}
		\includegraphics[width=\textwidth]{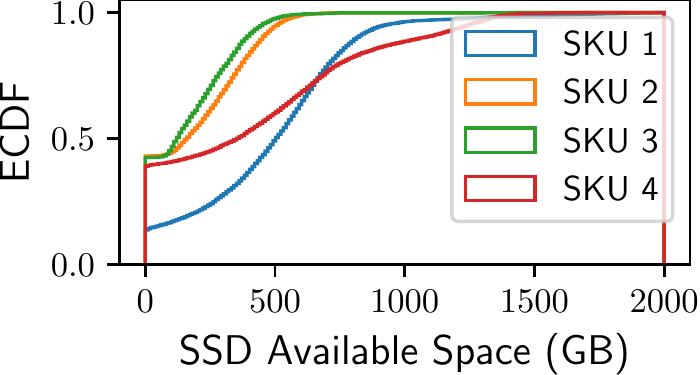}
	\end{minipage}%
	\begin{minipage}[b]{0.5\columnwidth}
		\includegraphics[width=\textwidth]{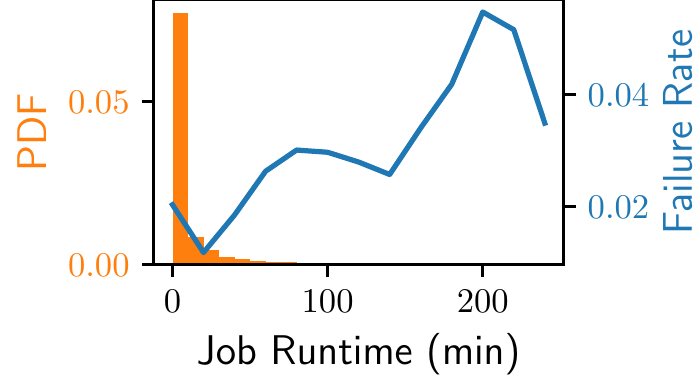}
	\end{minipage}%
	\vspace{-0.35cm}
	\caption{The Empirical Cumulative Density Function (ECDF) for SSD usage (left) and job failure rate with respect to job runtime and the probability density function (PDF) of job runtime distribution (right).}\label{fig:failure}\label{fig:ssd}
	\vspace{-0.4cm}
\end{figure}

\stitle{Introducing Phoebe.}
In this paper, we present Phoebe, a learning-based checkpointing optimizer for determining, at the compile-time, the decomposition (or the ``cuts'') of large job 
graphs in big data workloads.
\rev{Phoebe builds upon the state-of-the-art CLEO~\cite{siddiqui2020cost} cost models
and {\it fine-tunes} its operator-level predictions with historical statistics from past executions.
This is possible since production big data workloads are often recurrent, e.g., $>70\%$  in \cosmos~\cite{jyothi2016morpheus}.
Furthermore, checkpointing decisions only require cost predictions at stage boundaries (i.e., a set of operators that process a partition of data on a given node),
which are precise since stages get executed with physical boundaries%
, thus making the fine-tuning approach highly effective.
Phoebe applies a similar fine-tuning when predicting the time to live (TTL) for the output of each stage,
which is needed to estimate how long the data lives on temporary storage.
Phoebe uses a job runtime simulator and then fine-tunes its estimates with historical TTL of stage outputs.
}

\rev{Apart from cost models, Phoebe also introduces a scalable heuristics based checkpointing algorithm, that (1) can scale to millions of jobs in Cosmos workload; (2) it is two orders of magnitude faster than the optimal Integer Programming (IP) approach; and (3) yet strikingly close to the optimal in meeting the objective value.
Finally, to the best of our knowledge, Phoebe is the first checkpointing framework that supports multiple objectives and constraints, and hence could be used in several different scenarios. 
}

{\color{black}To summarize, we make the following key contributions:
\begin{itemize}
 \item We present Phoebe, a learning-based system that uses past workloads for making checkpoint decisions over future queries in big data workloads. (Section~\ref{sec:overview})
 \item We describe accurate stage-wise cost predictors for stage output size, stage runtime, and stage output TTL by fine-tuning over historical statistics seen in the past. (Section~\ref{sec:module1})
 \item  We introduce a scalable checkpoint optimization algorithm for large query DAGs and global constraints over hundreds of thousands of jobs, that can support several different checkpointing scenarios. (Section~\ref{sec:opt})
	\item We evaluate Phoebe over large production Cosmos workloads, and show how the various components contribute towards picking good checkpoints as well as the involved trade-offs. 
	Our results show that Phoebe can free more than $70\%$ of the local storage on hotspots, and reduce the recovery time for failed jobs by $64\%$ on average, while increasing job latency by less than 3\%. (Section~\ref{sec:result})
\end{itemize}
}

Below we first provide a background on Cosmos and SCOPE before presenting each of our contributions.

%% file: Motivation.tex
\section{Background}\label{sec:motivation}

\cosmos is the state-of-the-art big data analytics platform at \msft. 
It consists of hundreds of thousands of machines
executing hundreds of thousands of
jobs per day~\cite{ramakrishnan2017azure,shao2019griffon}.
\cosmos users submit their analytical jobs using \scope~\cite{chaiken2008scope,zhou2012scope},
a SQL-like data flow dialect. \scope{} jobs are compiled into a Direct Acyclic Graph (DAG) of stages which in turn are executed in parallel by a YARN-based scheduler~\cite{curino2019hydra}.

%% file: Background.tex
Figure~\ref{fig:graph} shows the execution plan for a $7$-stage \scope job, with each rectangle representing a stage. Within a stage, there can be multiple operators, such as 
\verb|Extract|, \verb|Filter|, 
etc., chained together. Each stage is packed into a {\it task} that runs in parallel 
on different data partitions on different machines. 
\begin{figure}[t]
	\centering
	\includegraphics[width=0.75\columnwidth]{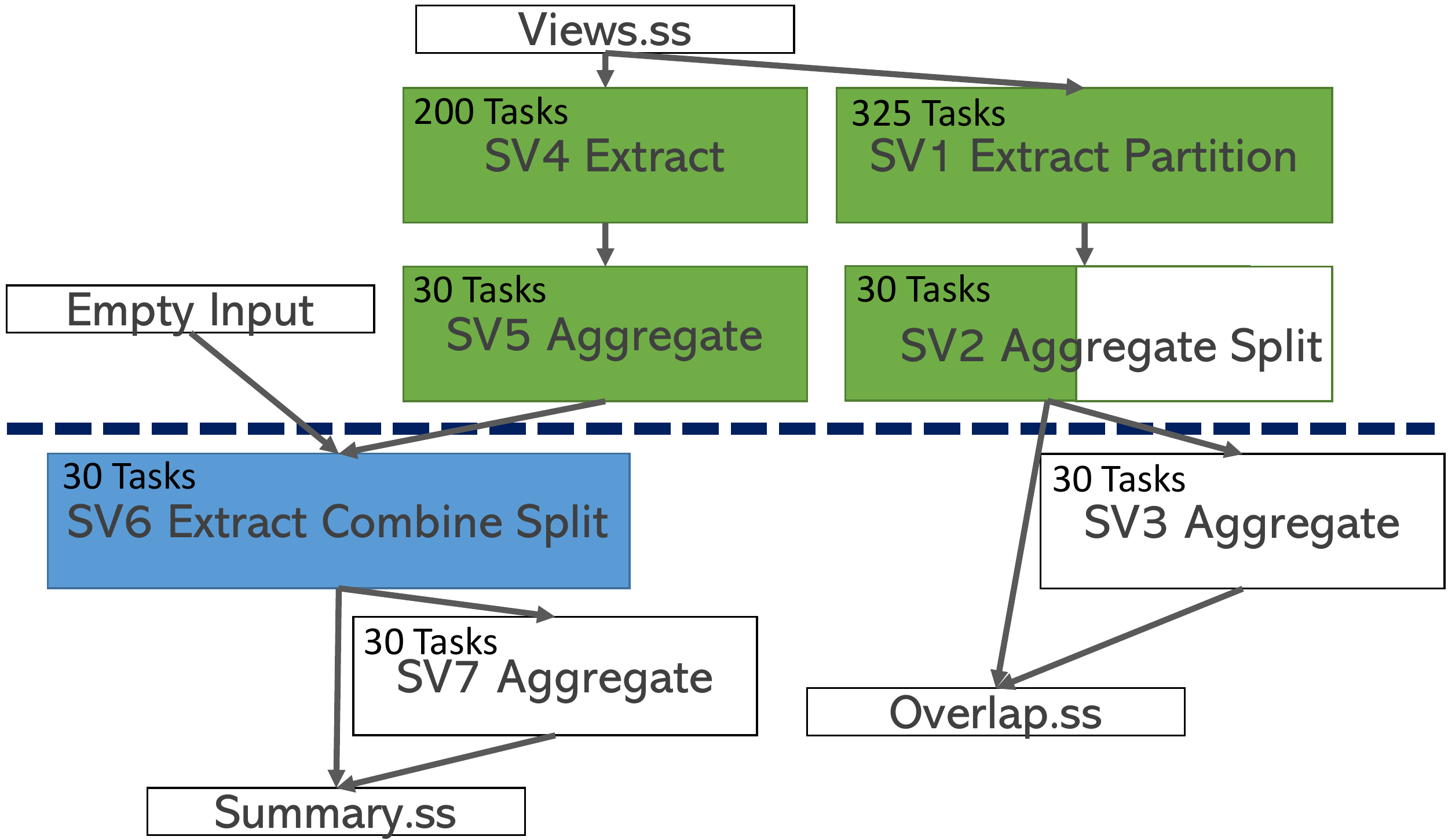}
	\vspace{-0.15cm}
	\caption{A job execution graph in \scope and one potential checkpoint decision (horizontal line).}\label{fig:graph}
	\vspace{-0.35cm}
\end{figure}
During the execution, users can monitor the progress of each stage (green means finished, blue means waiting, and white means not started). 
Based on the execution dependency, we call the dependee an \textit{upstream stage} and the dependent a \textit{downstream stage}. For instance, stage \verb|SV2_Aggregate_Split| in Figure~\ref{fig:graph} is an upstream stage of \verb|SV3_Aggregate| and a downstream stage of \verb|SV1_Extract_Partition|. An upstream stage usually (but not always) finishes before its downstream stage. 
When a stage finishes, its output will be saved to the local SSDs of each server, which we refer to as \textit{Temp Data Storage} in the rest of the paper.
\cosmos emits telemetry data recording not only the detailed execution plan for every single job, but also the schedule for every stage, its execution time and output/input size, type of operators involved, etc. 

Given the motivation to create persistent checkpoints, we want to 
decompose a job execution graph by selecting a set of stages for checkpointing, and redirecting their outputs to global HDD storage with 3 replicas.
We refer to these selected stages as {\it checkpoint stages}.
\emph{We want to select the checkpoint stages carefully such that the objective} (e.g., minimizing temp storage load) \emph{could be met, while constraining the global storage costs}. 
Finding the optimal checkpoint stages is similar to decomposing the execution plan and finding a cut in the graph. 
The dashed black line in Figure~\ref{fig:graph} illustrates an example cut in the execution graph for checkpoint selection. When we select Stages \verb|SV_2| and \verb|SV_5| as the checkpoint stages, 
we need to save the their outputs to global persistent stores. 
The space needed for the global store is {\color{black}proportional} to the sum of the output sizes of Stages \verb|SV_2| and \verb|SV_5|.

Objectives such as minimizing the overall temp data storage of all jobs, or minimizing the recovery/restart time for a failed job, 
depend on the time a job lives $t_u$ after each stage $u$, while the global storage constraints depend on the output size $o_u$ of each stage.
Furthermore, $t_u$, is a function of the runtime, $r_u$, of all stages in the execution graph, i.e., $t_u = f(r_1,r_2,...r_k)$. 
Accurate estimations of $o_u$, $r_u$, and $t_u$ are therefore crucial for a good checkpoint optimization.
In the following sections, we first present an overview of Phoebe, then we will describe the three ML models used for stage-wise costs ($o_u$ and $r_u$) and time-to-live predictions ($t_u$), respectively. %

%% file: Overview.tex
\section{Phoebe Overview}\label{sec:overview}

{\color{black} In this section we give an overview of Phoebe and highlight the design choices we have made.}
As discussed in the previous section, to determine the optimal cut(s) of an execution graph, 
it is important to estimate the output size, the runtime, and the time-to-live for each stage.
Unfortunately, the estimates the query optimizer in big data systems are off by orders of magnitude~\cite{microlearner}, due to 
(1) large query execution DAGs where the errors propagate exponentially~\cite{moerkotte2009preventing, salama2015cost}; (2)
prevalent use of custom {\color{black}user-defined functions that are hard to analyze~\cite{wu2018towards}; (3) recent works have exploited workload patterns to learn models for improving the cardinality estimates~\cite{wu2018towards,dutt2019selectivity,MCSN}, but still these learned estimates are not accurate enough in absolute values; and (4)}
the presence of both structured and unstructured input data~\cite{triou2020analyzing}.
{\color{black} \emph{\underline{Problem}: state-of-the-art cardinality estimation approaches are not good enough for predicting actual output sizes. \underline{Design choice}: Phoebe augments state-of-the-art learned cardinalities (i.e., CLEO~\cite{siddiqui2020cost}) by focusing on recurring jobs and exploiting historical statistics to instance-optimize the cardinality predictors.}}

Previous work on estimating cardinalities focus on improving the query optimizer estimates at the operator level. 
For checkpointing, however, we need to: {\color{black}{(1)}} estimate the costs at the stage-level, %
each consisting of multiple operators executing on a task in the same container; {\color{black}{(2)}} operators within a stage could be pipelined in different ways when scheduled on distributed tasks, which makes it non-trivial to combine individual operator costs into stage costs; {\color{black}{(3) stage outputs are persistent for the full duration of the job, therefore to estimate the storage costs we need to take into account this temporal dimension. 
\emph{\underline{Problem}: cardinality estimates at the operator level need to be aggregated at the stage level and augmented with a time dimension in order to properly model the storage cost. \underline{Design choice}: Phoebe generates stage-level estimates starting from the operator-level one, and adds a predictor for the time-to-live of each stage.}}}

{\color{black}
SCOPE-like big data engines have query plans that are DAGs of operators, not trees. Furthermore, Scope plans are complex: in our production workloads we have plans easily reaching thousands of operators. Prior works (e.g., MCSN~\cite{MCSN}, DeepDB~\cite{DBLP}, NEO~\cite{marcus2019neo}, NeuroCard~\cite{NeuroCard}, TBCNN~\cite{mou2016convolutional}, and~\cite{marcus2019plan,ortiz2018learning}) suggest to use DNNs to ``learn'' the encoding of relatively simple query structures and mapped each operator to neural unit(s). \emph{\underline{Problem}: Mapping Scope complex plans into deep neural networks results in severe gradient explosion or vanishing problems~\cite{resnet}. \underline{Design choice}: Phoebe captures the complex structure of big data query execution DAGs using a schedule simulator.}
}
Therefore, in this work, instead of a full black-box approach, we combine the existing work of cardinality estimation with an explainable simulation process, which is a judicious mixture of domain knowledge and principled data-science that leads to optimal results tailored to our complex production workloads. 

The checkpoint optimizer (Section~\ref{sec:opt}) uses the above estimates to make the checkpoint decisions.
 {\color{black}\emph{\underline{Problem}: Production checkpointing applications may have different objectives while the traditional checkpointing frameworks are rigid. \underline{Design choice}: Phoebe checkpointing algorithm is based on a ``graph cut'' algorithm that is adaptive to different objectives and constraints based on the specific application.}}

\begin{figure}[t]
	\centering
	\includegraphics[width=\columnwidth]{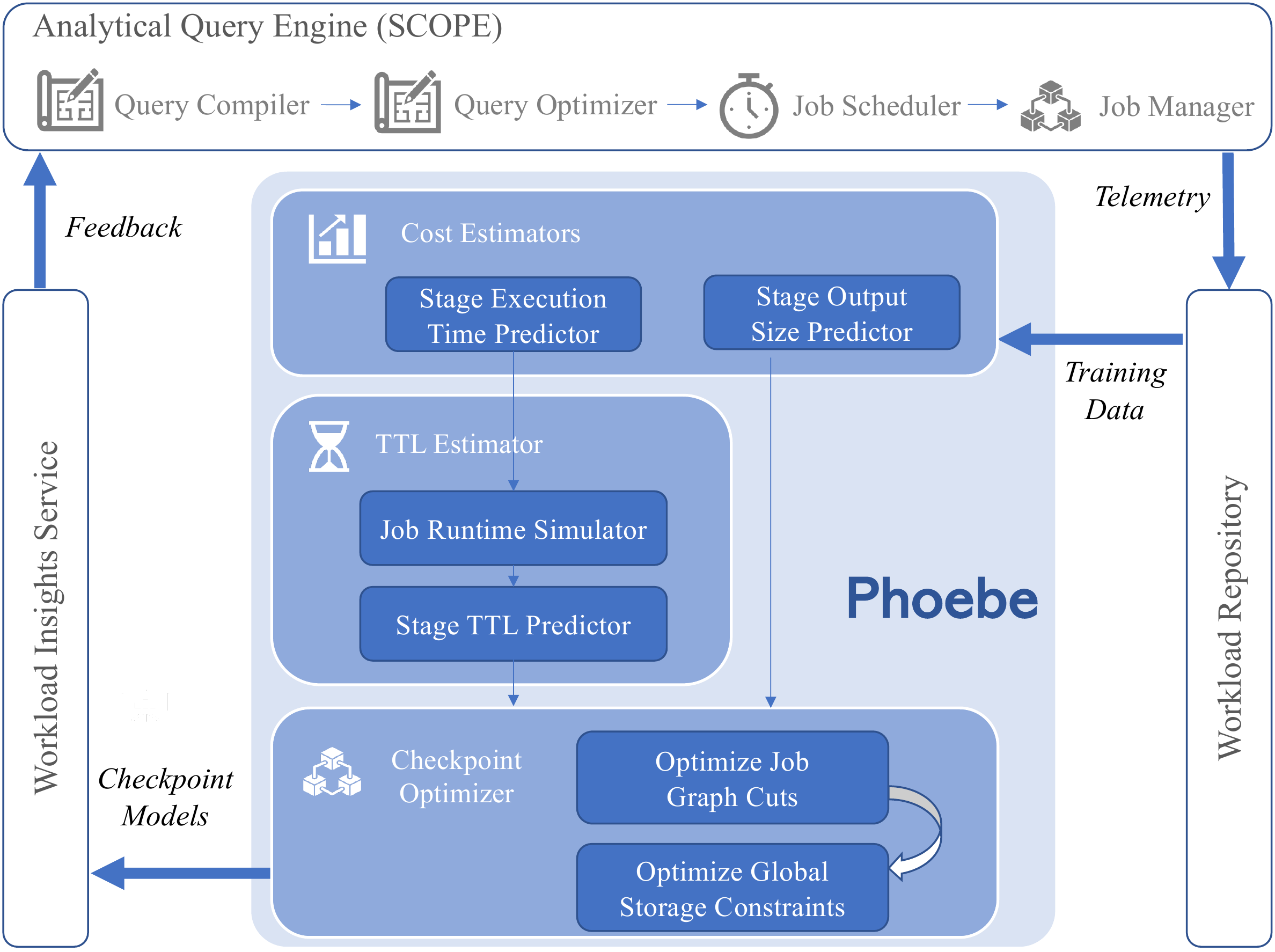}
	\vspace{-2ex}
	\caption{Phoebe architecture and its integration within the workload optimizer in \scope.}\label{fig:arch}
	\vspace{-0.6cm}
\end{figure}

Figure~\ref{fig:arch} shows the Phoebe architecture that is integrated with an already deployed workload optimization platform, namely Peregrine~\cite{jindal2019peregrine}. Phoebe consists of the following three modules:
\begin{packed_enum}
\item \textbf{The stage cost models} take as input the aggregated features at the stage level and uses machine learning methods to predict the duration of each stage, which is measured by the \textit{average} execution time for \textit{all} the tasks of the corresponding stage.
Likewise, we also learn models 
to predict the output size of each stage, i.e., the size of the output of the last operator in the stage. 

\item \textbf{The time-to-live (TTL) estimator} consists of two steps. First, a job runtime simulator infers the start/end times for each stage based on the job execution graph.
We assume a stage can only start once all its upstream stages have finished. TTL can be calculated as the time interval between the estimated stage end time to the estimated job end time. Similar to model stacking, we use a meta-learning model to further improve the TTL prediction from the simulator, i.e., we take the estimated TTL and the estimated time from start (TFS, defined as the time interval from the job start time to the start time of a stage) from the simulator, and use another machine learning model to generate the final TTL prediction.

\item \textbf{The checkpoint optimizer} uses as input the previous two modules, namely the estimated TTL and the output size of each stage, and selects the optimal set of global checkpoints given a particular objective function. To reduce the computation time for large workload sizes and to apply the storage constraints dynamically at runtime, we introduce a \rev{two-phase} approach to find the graph cuts and apply the storage constraint separately.
\end{packed_enum}

For a new job, the \scope compiler makes a call to the Workload Insight Service~\cite{jindal2019peregrine},
determines the checkpoint stages, modifies the query plan for materialization (similar as in CloudViews~\cite{jindal2018computation}),
and sends the information to the job manager, which takes care of checkpointing those stages to the global store. 
The telemetry data from the query engine is collected into a workload repository and later used by Phoebe to re-train the models.

%% file: CostModels.tex
\section{Stage-Level Predictors}\label{sec:module1}

In this section, we discuss the stage-level predictors,
namely the predictors for execution time and output size (Section~\ref{sec:output-stage})
and the predictor for TTL (Section~\ref{sec:ttl}).

\eat{
 that model the stage output size (Section~\ref{sec:output-stage}), and the execution time (Section~\ref{sec:ttl}), respectively. %
Both predictors use the same model structure and are based on two sets of features: (1) the cardinality estimates from the state-of-the-art query optimizer at the operator level~\cite{jindal2018computation,jindal2019peregrine,siddiqui2020cost}; and (2) the historic statistics of repetitive occurrences of a job template. %
{\color{black}This latter set of features allow to leverage past workload telemetry to calibrate the models, therefore providing instance-optimized models.
We developed both a traditional machine learning (ML) models as well as a Deep Neural Net (DNN) model as benchmark. }
Below we first detail the features used for our predictors (Section~4.1); afterwards we will describe the models (Sections~4.2 and 4.3).
}

\subsection{Execution Time \& Output Size Estimators}\label{sec:output-stage}

\subsubsection{Input Features}

{\color{black} 
Cosmos implements a state-of-the-art query optimizer and learned cost models, CLEO~\cite{siddiqui2020cost}. 
CLEO generates a collection of cost models, one for each common sub-graph in the plans. Each sub-graph corresponds to the same root physical operator (e.g.,
	\verb*|Filter|) and all upstream operators (e.g., \verb*|Scan|). The rationale is that cloud workloads are quite diverse in nature, and historically the one-fits-all models have failed to improve the estimated costs. CLEO~\cite{siddiqui2020cost} has proved to have 2 to 3 orders of magnitude more accurate than existing approaches. 
In this work, we leverage CLEO and extend it in three ways. }

{\color{black} 
\begin{packed_enum}
\item We use CLEO operator-level features as input to generate stage-level estimates. Stage-level estimates do not correspond to any sub-graphs, but they are estimates of all the operators combined into a stage by the SCOPE optimizer. These input features are directly accessible from the SCOPE optimizer. 
\item We use historical data coming from the previous occurrences of the recurrent job to instance-optimize the predictors.
In Cosmos, even with a large number of recurrent jobs, the parameters, inputs and execution plan can vary significantly over time. Therefore, it is important to not only capture repetitive patterns but also leverage the specific context of each of the stages in the workload.
\item Finally, the input file paths and the job names often preserve information of file type, or locations and can be used as text features. For instance, a log file with file names including ``\verb|log|'' usually consists of raw text in string format, which makes it more time-consuming to process compared to input with an ending of \verb|*.ss| (structured steam~\cite{triou2020analyzing}, a SCOPE internal file format).
\end{packed_enum}
}

\rev{In summary, we constructed three groups of features  as shown in Table~\ref{tab:features}.
As we will see in Section~\ref{sec:result_cost}, it is the combination of these input feature sets that yield the best prediction accuracy.}

\begin{table}[t]
	\small
	\caption{Cost model features}\label{tab:features}\vspace{-2ex}
	\begin{tabular}{p{2.1cm}p{3.5cm}p{1.7cm}}
		\toprule
		Feature Group & Feature Name     & Feature Description                          \\ \midrule
		Query Optimizer Features & \textit{Estimated Cost, Estimated Input Cardinality, Estimated Exclusive Cost, Estimated Cardinality} for the last operator of the stage & Numeric features from the optimizer's internal information\\
		Historic Statistics      & the \textit{Exclusive Time} and the \textit{Output Size}  for the job template and operator combination & The historic average of the statistics\\
		Normalized File Path/Job Name      & \verb|Norm| \verb|Job| \verb|Name|, \verb|Norm| \verb|Input| \verb|Name| & Text features                               \\ \bottomrule
	\end{tabular}
\vspace{-0.25cm}
\end{table}

\subsubsection{Model Implementation}

NimbusML~\cite{nimbusml} is an open-source python package for ML.NET~\cite{ahmed2019machine}. %
We tried different ML learners from NimbusML (e.g., linear regression, ensemble regression, etc.) and found that the LightGBM learner~\cite{ke2017lightgbm} is the best in terms of prediction accuracy for our use case. We developed two groups of models, each of them using only the non-textual features (i.e., query optimizer features and historic statistics): {\color{black}(1) \textbf{General model}, i.e., one model for all stages; and (2) \textbf{Stage-type specific model}. In fact, we observe that stages can be divided by their \emph{type} based on the operators involved. Similar to CLEO~\cite{siddiqui2020cost} where the model is sub-graph specific, the stage-type corresponds to a unique set of (usually one or two) operators forming the stage, e.g., an \verb*|Extract_Split| stage has a \verb*|Process| operator followed by a \verb*|Split| operator. %
In the production workload used in this paper, we observed 33 stage types. We therefore train stage-type specific models, each with more homogeneous data. The stage-type specific models capture the heterogeneity of runtime variation across different combinations of operators. Given that we only select recurrent jobs for the checkpoint mechanism, it is desirable for the cost model to ``overfitted'' the selected recurrent jobs.}

{\color{black}To leverage text features such as \verb|Norm| \verb|Input| \verb|Name| and the \verb|Norm| \verb|Job| \verb|Name| where simple One Hot Encoding~\cite{bishop2006pattern} is not possible, we trained a customized word embedding using a language model~\cite{bengio2003neural} and integrated it with another DNN model with 2 hidden layers to predict the final targets as a benchmark. }
We host the end-to-end model training process on Azure ML for better experiment tracking and model archiving~\cite{aml}.

%% file: TTL.tex
\subsection{Time-to-Live Estimator}\label{sec:ttl}

The time-to-live (TTL) estimator predicts the average lifetime of the intermediate output of each stage, defined as the time interval from the average end time of \textit{all} tasks in the corresponding stage to the end time of the job. {\color{black}This is different from the estimation at the sub-tree/sub-query level because the TTL can be impacted by operators not included in the sub-tree which determine the job end time.} Instead of training a DNN model to capture the complex dependency structure between stages (as in~\cite{marcus2019plan,marcus2019neo,ortiz2018learning}), we introduce a simple schedule simulator to mimic the job execution process in Cosmos.
The TTL estimator consists of two steps. First, a job runtime simulator takes as input (1) the stage execution time (i.e., the average task latency) estimated by the stage execution time predictor from the previous section; 
and (2) the execution graph to simulate the job execution process. Second, we develop another machine learning model to further improve the TTL prediction based on the simulator output.
In the following sections, we discuss each step in more detail.  

\subsubsection{Job Runtime Simulator}

The job runtime simulator estimates the start and end time of each stage based on the predictions of the stage execution time and the dependency relationship in the execution graph. To simplify the modeling, we assume strict stage boundaries, i.e., each stage can only start after all of its upstream stages have finished.
A topological sorting~\cite{sort} algorithm sorts all stages in a linear order based on the execution graph\footnote{This is very similar to how the SCOPE job manager schedules tasks in Cosmos.}, such that an upstream stage executes before a downstream stage. The schedule simulator uses the linear ordering of the stages to estimate the stages' start and end times. For each stage, the simulator calculates its start time based on the maximum end time of all its upstream stages, and estimates its end time based on the estimated stage execution time from the stage execution time predictor.

Algorithm~\ref{algo:simulator} shows the detailed schedule simulator process. 
Based on the linear ordering from the topological sorting algorithm, we schedule stages sequentially from the front of the ordered stack (from position 0). 
\begin{algorithm}[t] %
	\small
	\DontPrintSemicolon
	\SetKwInOut{Input}{Input}
	\SetKwInOut{Output}{Output}
	\SetKwInOut{Initialize}{Initialize}
	\SetKwFunction{UpstreamStages}{UpstreamStages}
	\SetKwFunction{upstream}{upstream}
	\SetKwFunction{MaxUpstreamEndTime}{MaxUpstreamEndTime}
	\SetKwFunction{ScheduleSimulator}{ScheduleSimulator}
	\SetKwFunction{True}{True}
	\SetKwFunction{False}{False}
	
	\Input{execution graph $G$, ordered stack $R$, estimated execution time $T$}
	\Output{start time for stages $D[s]$ \\ end time for stages $P[s]$}
	\Initialize{$D[s]=$Null, $\forall s \in R$ \\ $P[s]=$Null, $\forall s \in R$}
	\BlankLine	
	\ForEach{stage $s \in R$} {
		\MaxUpstreamEndTime = 0\;
		\If{$s$.\UpstreamStages != Null} {
			\ForEach{$\upstream \in s.\UpstreamStages$} {
				\MaxUpstreamEndTime = $\max \big\{\MaxUpstreamEndTime, P[\upstream]\big\}$ 
			}
		}
		$D[s] = \MaxUpstreamEndTime$\;
		$P[s] = D[s] + T[s]$\;
	}	
	\Return $D,P$
	\caption{\texttt{ScheduleSimulator}} 
	\label{algo:simulator} 
\end{algorithm}	
The TTL can be calculated as the time interval between the stage end time and the job end time.

\subsubsection{Fine-tuning}

{\color{black}While the simulator assumes a strict stage boundary and captures the dependency between stages, it doesn't simulate the pipelined operation. In \cosmos, for some stage types, tasks can start before all the tasks of upstream stages finish. Strict stage boundary assumption is helpful for computation efficiency; however, it potentially results in overestimating the TTL. Therefore, we create an ML model to systematically adjust for this bias by stage-type. We observed that some of the stage types usually have longer or shorter TTL, such as \verb|Extract|, that always starts before all the other stages thus with longer TTL or an \verb*|Aggregate| stage that tends to be placed towards the end of the job thus has shorter TTL. Therefore, we develop machine learning models per stage-type to have different adjustment mechanisms to achieve better accuracy. }
The input feature for the stacking model includes the estimated TTL from the simulator as well as the time from start (TFS), which is defined as the time interval between the start time of the job to the start time of the corresponding stage. Those two values define the ``position'' of this stage throughout the execution of the job. %

%% file: Optimizer.tex
\begin{figure}[t]
	\centering
	\vspace{-0.25ex}
		\includegraphics[width=0.94\columnwidth]{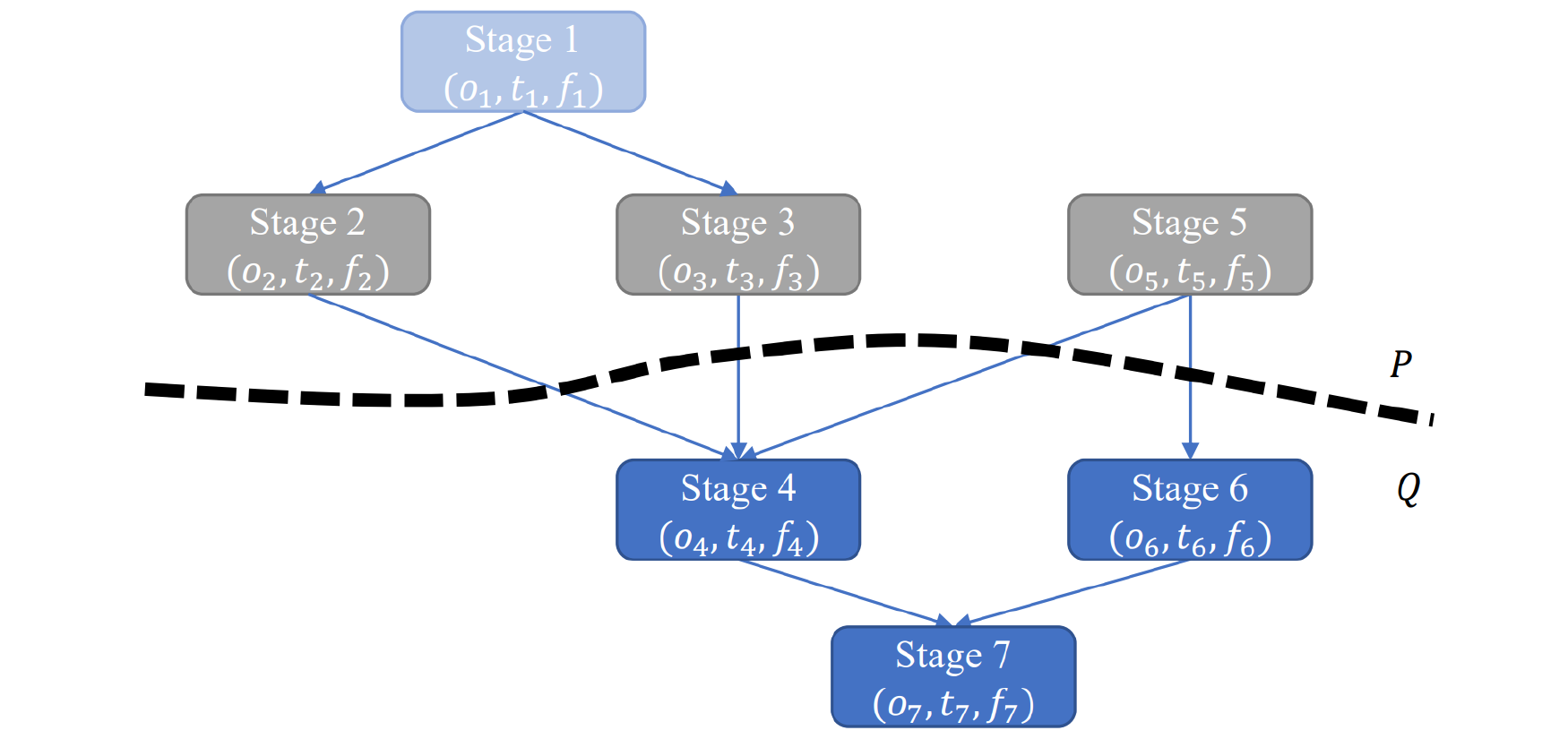}
		\vspace{-0.2cm}
		\caption{Graph cut in the integer programming.}\label{fig:ip}
		\vspace{-0.4cm}
\end{figure}

\section{Checkpoint Optimizer}\label{sec:opt}

We now describe the checkpoint optimizer. {\color{black}Similar to Flint~\cite{sharma2016flint}, we only consider a set of ``frontiers" of the program's lineage graph. 
The checkpointing problem can then be naturally mapped to finding a cut in the execution graph.
}
We categorize stages in a job execution graph into three groups ({\color{black}with respect to each cut):} 
\begin{enumerate}
\item Group I: the checkpoint stages, i.e., stages that need to persist their outputs to the global storage;
\item Group II: stages that have finished executing before the checkpoint stages; and 
\item Group III: stages that will execute after the checkpoint stages. 
\end{enumerate}

Phoebe's checkpointing optimizer is extensible: it can be tuned using different objective functions based on different checkpointing applications, as described earlier in Section~\ref{sec:intro}. 
In particular, we discuss two of the applications in this section, namely freeing up temp data storage on hotspots and quickly restarting failed jobs.
Below, we first show an IP formulation of the single-cut checkpoint problem. 
Then we show how it can be extended for multiple cuts. 
Finally, to solve the IP efficiently, we split our formulation into two sub-problems: (1)~a heuristic approach to {\color{black}obtain efficiently the optimal solution} without considering the global storage cost and with one cut; and (2)~enforcing the global storage constraints.

\subsection{Integer Programming}  \label{sec:ip} 

{\color{black}
Table~\ref{tab:notation} summarizes the notation.
\begin{table}[t]
	\small
	\caption{{\color{black}Notation}}\label{tab:notation}
	\vspace{-2ex}
		\begin{tabular}{cp{6.5cm}}
			\toprule
			  variable&description\\\midrule
			  $u$ & stage index \\
              $S$ & set of stages \\
            	 $E$ & set of edges\\      
               	 $o_u$ & output size of stage $u$ \\       
               	 $t_u$ & time-to-live of stage $u$ \\		                    	
          		$G$& auxiliary variable, the total global storage usage\\
           		$T$& auxiliary variable, the total saving for temp data storage\\ 
           		$\alpha$& cost factor of using global storage\\      		
           		\midrule
            	 $z_u$ & decision variable representing if stage $u$ is before the cut \\		                    
				$d_{uv}$ & decision variable representing if edge $(u,v)$ is on the cut  \\
				$g_u$& auxiliary variable representing if $u$ is a checkpoint stage\\   \midrule
			 $\mathcal{C}$& set of cuts\\ 
			 $c$& cut index\\   
			  $z_u^{(c)}$ & decision variable representing if $u$ is before the cut $c$\\		                    
			 $d_{uv}^{(c)}$ & decision variable telling if edge $(u,v)$ is on the cut $c$ \\\bottomrule
	\end{tabular}
	\vspace{-2.5ex}
\end{table}
}
Let us consider the case of adding one cut in the execution graph. For stage $u$, assume its output size $o_u (\geq 0)$ and time-to-live (TTL) $t_u (\geq 0)$
 are known (see Figure~\ref{fig:ip}). Let $z_u,~\forall u \in S$ be a set of binary decision variables representing if stage $u$ is executed before the graph cut, i.e., on the $P$-side of the cut (e.g., $z_u=1$ indicates that stage $u$ is before the cut). Let $d_{uv},~\forall (u,v) \in E$ be another set of decision variables representing if edge $(u,v)$ is on the cut, where $E$ denotes the set of edges in this execution plan. For instance, in Figure~\ref{fig:ip}, $d_{24}$,~$d_{34}$,~$d_{54}$ and $d_{56}=1$. The total number of decision variables is $|S| + |E|$, and the number of possible combinations is $2^{|S| + |E|}$ as they are all binary. 

We can now use the two sets of binary decision variables to decide the grouping of a given stage $u$:
\begin{itemize}
	\item Group I: $z_u=1$ and $d_{uv}=0,~ \forall (u,v) \in E$,
	\item Group II: $z_u = 1$ and $\exists d_{uv}=1,~ \forall (u,v) \in E$, and 
	\item Group III: $z_u = 0$.
\end{itemize}

The space needed for global storage is {\color{black}proportional to }the sum of the total output sizes for stages in Group II. 
The checkpoint optimization problem can be formulated as an IP as follows:

\vspace{-0.3cm}
{\small
\begin{align}
\max_{\substack{z_u, \forall u\in S\\ d_{uv},\forall (u,v)\in E}}&  \text{Obj}(z_u,d_{uv}) \label{eq:01}\\
\text{s.t.}~ &G=\sum_{u\in S}o_ug_u,\label{eq:02}\\
& g_u \geq d_{uv},~\forall (u,v)\in E,~\forall u,\label{eq:03}\\
& d_{uv}-z_u+z_v\geq0,~\forall (u,v)\in E,\label{eq:07}\\
& d_{uv} \in\{0,1\},~\forall (u,v)\in E,\label{eq:08}\\
& z_u\in\{0,1\},~\forall u \in S\label{eq:09}.
\qedhere
\end{align}
}
Equation ~\eqref{eq:01} can be replaced by different objective functions depending on the application. We will discuss two possible applications in the following sections. 
Constraint~\eqref{eq:02} calculates the total global storage needed from stages in Group II. It can be estimated by examining the output size on the upstream side of the edges on the cut. In Constraint~\eqref{eq:03}, $g_u=1$ for stages that have any edges on the cut ($d_{uv}=1$). Note that, this requires that the objective function minimizes $G$ and $g_u$. 
Constraint~\eqref{eq:07} ensures that if stage $u$ is on the side before the cut (the side of $P$) and $v$ is after the cut, $d_{uv}=1$.

\vspace{0.2cm}
\noindent\textbf{Multiple Cuts.} Let $z^{(c)}_u$ denote the binary decision variable indicating whether stage $u$ is before the cut $c \in \mathcal{C}=\{0,1,\cdots, K\}$, $K\geq 1$. Let $d_{uv}^{(c)}$ denote if edge $(u,v)$ is on the cut $c$. The optimization problem for multi-cuts can be formulated as follows:
{\small
\begin{align}
\max_{\substack{z_u^{(c)}, \forall u\in S,~c \in \mathcal{C}\\d_{uv}^{(c)}, \forall (u,v)\in E, ~c\in \mathcal{C}}}&  \text{Obj}\left(z_u^{(c)},d_{uv}^{(c)}\right) \label{eq:a1}\\
\text{s.t.}~ &G=\sum_{u\in S}o_ug_u,\label{eq:a2}\\
&g_u \geq d_{uv}^{(c)},~\forall (u,v)\in E,~\forall c \in \{1,\cdots, K\},~\forall u, \label{eq:a22}\\
& z_u^{(c-1)} \leq~z_u^{(c)},~\forall c \in \{1,\cdots, K\}\label{eq:a08}\\
&d_{uv}^{(c)}-z_u^{(c)}+z_v^{(c)}\geq0,~\forall (u,v)\in E,~\forall c \in \mathcal{C},\label{eq:a7}\\
& \sum_{c \in \mathcal{C}} d_{uv}^{(c)} \leq~1~\forall (u,v)\in E\label{eq:a9}\\
& d_{uv}^{(c)} \in\{0,1\},~\forall (u,v)\in E,~\forall c \in \mathcal{C},\\
& z_u^{(c)}\in\{0,1\},~\forall u \in S,~\forall c \in \mathcal{C}\label{eq:a12}.
\qedhere
\end{align}}
The total number of decision variables is $(K+1)\cdot(|S| + |E|)$, and the number of possible combinations is $2^{(K+1)\cdot(|S| + |E|)}$. 
Constraint~\eqref{eq:a22} is similar to~\eqref{eq:03} and requires that the objective function aims to minimize $G$.
Constraint~\eqref{eq:a08} ensures that the cut $c-1$ comes ``before'' the cut $c$, i.e., all the stages before the cut $c-1$ will also be before the cut $c$. 
Constraint~\eqref{eq:a7} is similar \rev{to} the single-cut formulation, and Constraint~\eqref{eq:a9} ensures that the cuts do not overlap.

\subsection{Freeing Temp Data Storage in Hotspots}

{\color{black}With checkpointing, temp data for stages before the cut can be cleared when all the checkpoint stages finish as opposed to the end of the job. And this time difference is equal to the shortest TTL among those stages (the last one to finish).} For saving temp data storage on hotspots, we maximize the saving calculated as the product of the shortest TTL and the sum of the total output sizes for stages in Group I and II, $\sum_{v\in S: z_v=1}o_v \min_{u: z_u=1}t_u$. Considering the cost factor of using global storage to be $\alpha$, the IP formulation is as follows:
{\small
\begin{align}
\max_{\substack{z_u, \forall u\in S\\ d_{uv},\forall (u,v)\in E}}&  T-\alpha G \label{eq:11}\\
\text{s.t.}~ &
T=\sum_{u\in S}o_uw_u,\label{eq:33}\\
&w_u\leq t + M(1-z_u),~\forall u \in S,\label{eq:44}\\
&w_u\leq Mz_u,~\forall u \in S,\label{eq:55}\\
&t\leq t_u + M(1-z_u),~\forall u \in S,\label{eq:66}\\
&\eqref{eq:02} - \eqref{eq:09}, \nonumber
\qedhere
\end{align}}
where, $T$ calculates the total saving for temp data storage and $M$ is a sufficiently large number. Given Constraints~\eqref{eq:44} and~\eqref{eq:55}, $w_u=t$ if $z_u=1$ else 0 in the maximization problem. Therefore, $w_u=0$ for stages after the cut with $z_u=0$.
Similarly, Constraint~\eqref{eq:66} calculates the minimum TTL for stages before the cut in this maximization problem. $t\leq t_u$ when $z_u=1$ and $t\leq M$ when $z_u=0$, which can be ignored. This set of constraints ensure that in the maximization problem, $t$ is equal to the minimum of $t_u$ $\forall u$ where $z_u=1$.

For the cases with multiple cuts, the formulation is as follows:
{\small
\begin{align}
\max_{\substack{z_u^{(c)}, \forall u\in S,~c \in \mathcal{C}}}&  T - \alpha G \label{eq:a11}\\
\text{s.t.}~ 
&T=\sum_{u\in S}o_u\sum_{c \in \mathcal{C}}w_u^{(c)},\label{eq:a3}\\
&\Delta z_u^{(0)} = z_u^{(0)},\label{eq:at11}\\
&\Delta z_u^{(c)} = z_u^{(c)} -  z_u^{(c-1)},~\forall c \in \{1,\cdots, K\},\label{eq:at22}\\
&w_u^{(c)}\leq t^{(c)} + M(1-\Delta z_u^{(c)}),~\forall u \in S, \forall c \in \mathcal{C},\label{eq:a44}\\
&w_u^{(c)}\leq M\Delta z_u^{(c)},~\forall u \in S,~\forall c \in \mathcal{C},\label{eq:a55}\\
&t^{(c)}\leq t_u + M(1-z_u^{(c)}),~\forall u \in S,~\forall c \in \mathcal{C},\label{eq:a66} \\
&\eqref{eq:a2} - \eqref{eq:a12}. \nonumber
\qedhere
\end{align}}
Constraints~\eqref{eq:at11} and~\eqref{eq:at22} introduce $\Delta z_u^{(0)}$, indicating if stage $u$ is before the cut $0$, and $\Delta z_u^{(c)}~\forall c \in \{1,\cdots, K\}$ indicating if stage $u$ is between the cuts $c-1$ and $c$. 
Given Constraints~\eqref{eq:a44} and~\eqref{eq:a55}, $w_u^{(c)}=t^{(c)}$ for stages with $\Delta z_u^{(c)}=1$, else $0$ in this maximization problem.
$t^{(c)}$ in Constraint~\eqref{eq:a66} calculates the minimum TTL for stages before the cut $c$. $t^{(c)} \leq t_u$ when $z_u^{(c)}=1$, and $t^{(c)} \leq M$ when  $z_u^{(c)}=0$. In the maximization problem, $t^{(c)}$ is equal to the minimum of $t_u~\forall u$ where $z_u^{(c)}=1$.
Combining Constraints~\eqref{eq:a44},~\eqref{eq:a55} and~\eqref{eq:a66}, for each stage, we calculate the corresponding temp data saving based on the minimum TTL for stages in the same group who are between the same pair of cuts.
Since all constraints in the formulation are linear, the IP can be solved by existing solvers.

Solving the above IP over large execution graphs can take long, and given that we need to solve the IP for every job, it can quickly become operationally expensive. Therefore, we propose a heuristic-based solution to solve the IP formulation from the previous section at interactive speed, i.e., it can be run during job execution. The key idea is to maximize the objective and apply the global storage constraint separately, i.e., ignore the cost of using global storage when determining the cuts and consider adding only one cut. As a result, the IP formulation reduces to have only one set of decision variables, $z_u~\forall u \in S$.
{\small
\begin{align}
(\texttt{OptCheck1})\quad&\max_{\substack{z_u, \forall u\in S}}  T \label{eq:1s}\\
\text{s.t.}~ 
&\eqref{eq:09},~\eqref{eq:33} - \eqref{eq:66}. \nonumber
\qedhere
\end{align}
}
The above reduction has the following interesting property.

\begin{proposition}
	For any model primitives $t_u$ and $o_u$, there exists optimal solutions $z^*$ of problem \verb|OptCheck1| such that there exists $v\in S$ such that $z^*_u=1, \forall u\in S: t_u\ge t_v$ and  $z^*_u=0, \forall u\in S: t_u<t_v$. 
\end{proposition}
\begin{proof}
	We prove by contradiction, suppose there exists an optimal solution $z'$ such that there exists $u,u'\in S$, such that $z'_{u}=0, z'_{u'}=1$ and $t_{u}>t_{u'}$. Constructing a new solution $z^*$ as follows:{\small
	\begin{align}
	z^*_u &= 1,\\
	z^*_{v} &= z'_{v}, \forall v\neq u.\\
	\sum_{v\in S: z^*_v=1}o_v \min_{u: z^*_u=1}t_u = \sum_{v\in S: z^*_v=1}&o_v \min_{u: z'_u=1}t_u \ge\sum_{v\in S: z'_v=1}o_v \min_{u: z'_u=1}t_u.
	\end{align}}	
	This indicates that for any optimal solutions $z'$ that does not satisfy the conditions stated in the proposition, we can always find another one $z^*$ with at least the same objective value. %
\end{proof}

Therefore, \textbf{we can enumerate the different combinations for $z_u$ for all the stages in an efficient way}. Specifically, we start from the first stage (ordered by TTL) and set its $z_u$ to 1 and the others to 0. Then, we incrementally add more stages to the side before the cut, set their $z_u$ to 1, and evaluate the corresponding objective value for each new solution. 
This is equivalent to finding the optimal timestamp before which all stages have $z_u=1$ and the total number of combinations is $|S|$. 
We consider the global storage cost separately and discuss that in Section~\ref{sec:capa}. 

This enumeration process is illustrated in Figure~\ref{fig:cumu}. Considering the execution process of a query job, after each stage finishes, the corresponding output will be saved to the temp data storage. Therefore, the total size of the temp data storage being used is increasing monotonically in time with more stages finishing. On the other hand, assuming that at the end time of stage $s$, we decide to clear the temp data storage for all the previous stages, the resulting saving for the temp data storage is equal to the TTL for this stage multiplied by the current temp storage usage size. The global storage space needed for Group III stages can be inferred based on the current execution status for each stage. 

\begin{figure}[t]
	\centering
	\includegraphics[width=0.98\columnwidth]{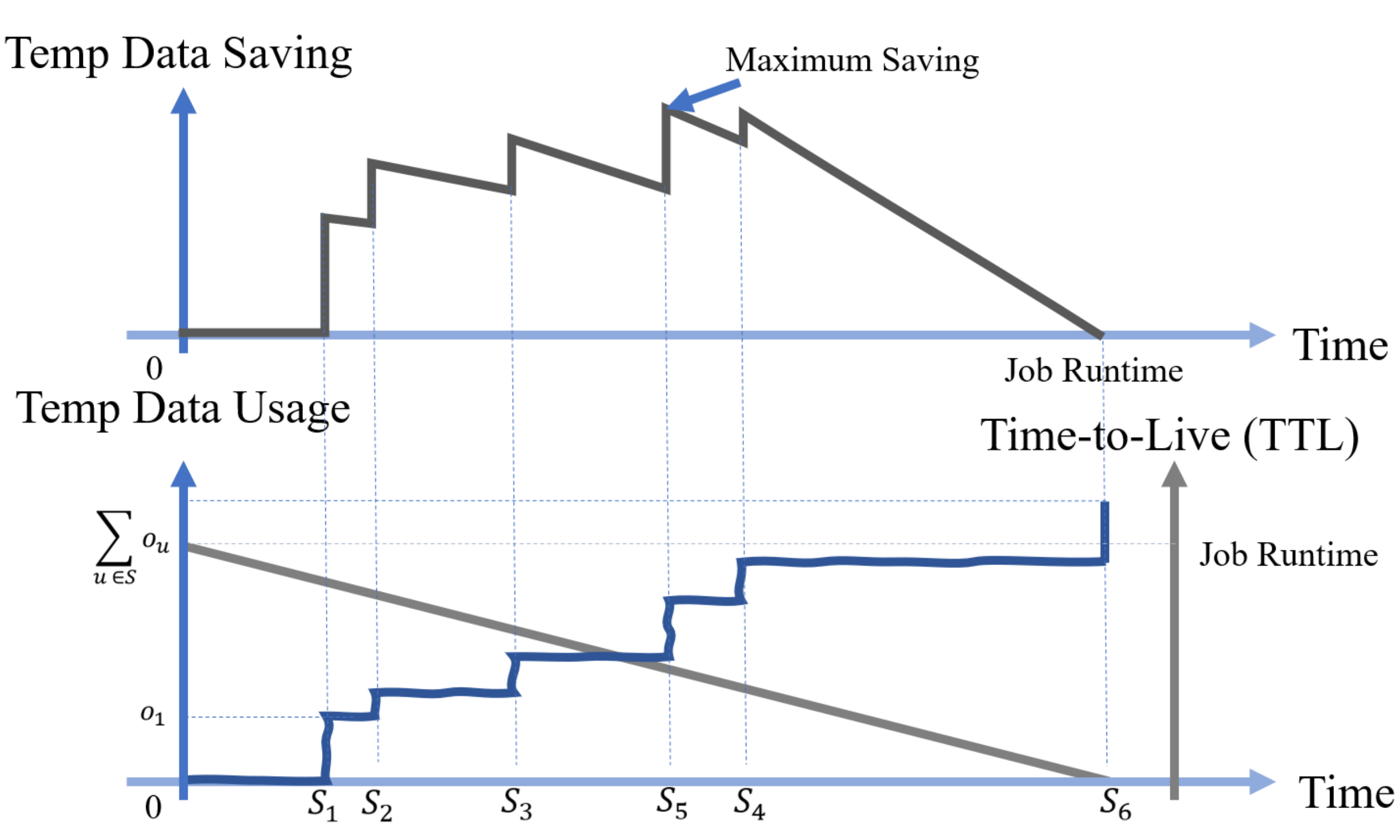}
	\vspace{-2.0ex}
	\caption{Potential temp data saving as a function of time.}\label{fig:cumu}
	\vspace{-2.0ex}
\end{figure}

\subsection{Restarting Failed Jobs}\label{sec:ext}
{\color{black}Some approaches~\cite{sharma2016flint,salama2015cost,daly2006higher} assume a distribution of the inter-arrival time for failures/interrupts given the mean time between failures/interrupts (MTBF/MTBI) as parameters. Others~\cite{chen2013selective, upadhyaya2011latency,yang2010osprey} assume a fixed probability of failure for an operator running on a worker node. %
In Cosmos, the execution time of tasks (30-40 seconds on average) has much less variation compared to the total job runtime (ranging from seconds to days). Therefore we find both the two assumptions above to hold. 

Let $\delta$ denote the average failure rate for a task, it can be approximated by the average task runtime and MTBF, given that the average runtime is much shorter than MTBF:
{\small
\begin{align}
	\delta\approx~\mathbb{E}(\text{Task Runtime})/MTBF.
\end{align}}
For Cosmos, $\delta\ll0.05$, meaning that MTBF is in the order of hours. The probability of failure for stage $u$ with $v_u$ tasks is equal to the probability of having no tasks to fail, which follows a binomial distribution:
{\small
\begin{align}
1-(1-\delta)^{v_u}\approx \delta v_u.
\end{align}}
}
Let $\bar{t_u}$ denote the time from start (TFS) for stage $u$, indicating the time interval from the job start time to the stage start time. The IP formulation for maximizing the recovery time saving with a single cut and not consider the global storage is as follows:
{\small
\begin{align}
(\texttt{OptCheck2})&\quad\max_{\substack{z_u, \forall u\in S}}  P_F\bar{T} \label{eq:f1s}\\
\text{s.t.}~ \bar{T}&\leq \bar{t_u}(1-z_u) + Mz_u\label{eq:3fs},\\
P_F&=\prod_{u \in S}^{}(1-\delta v_u)^{z_u}\left(1 - \prod_{u' \in S}^{}(1-\delta v_{u'})^{1-z_{u'}}\right), \label{eq:resiliency}\\
z_u&\in\{0,1\},~\forall u \in S,
\qedhere
\end{align}}
where Constraint~\eqref{eq:3fs} calculates the minimum TFS for all stages in Group III, i.e., $z_u=1$. Constraint~\eqref{eq:resiliency} calculates the probability of failure in one of the stages on the side after the cut but not before.
One can prove the optimal solution of the problem \verb|OptCheck2| has the similar property as in \verb|OptCheck1|. Therefore, the same heuristic solution can be applied to incrementally add stages to the side before the cut and search for the optimal checkpoint time. 
To maximize the expected time saving for the checkpoint mechanism considering the probabilities of stage failures, the trade-off is between a higher failure probability and a larger saving for the restarting time. 
A figure similar to Figure~\ref{fig:cumu} can be generated to show the changes of the probability of failing after a stage and the corresponding recovery time saving estimated based on the time from start (TFS) for the corresponding checkpoint stages, as a function of time.

\subsection{Capacity Constraints on Global Storage}\label{sec:capa}
\rev{The global storage is a separate storage system where the data is 3x replicated and, for operational reasons, it is cleaned regularly (e.g., every 7 days). At SCOPE scale, this translates in a capacity in the order of a few PBs per day.}
In the previous sections we did not consider the cost of using the global storage. However, we can incorporate global storage constraints when selecting the \textit{jobs} for checkpointing. Given the hundreds of thousands of jobs running every day, we can only checkpoint a fraction of them due to the checkpoint costs and overheads involved. 
Therefore, based on the objective value and the space needed for global storage for each job, we can be more selective about which job to checkpoint and achieve a high cost-benefit ratio. {\color{black}For instance, as shown in Figure~\ref{fig:failure}, long-running jobs are more likely to fail due to the large number of tasks. Thus checkpointing can be enabled only for long-running jobs.}

Consider a time period $T$ (e.g., 1 day) and let the checkpoint budget for global storage be $W$. Let $w_i$ denote the global storage space needed and $\pi_i$ denote the ratio between the objective value and the global storage needed for job $i$. The problem of applying global storage constraint can be seen as an online stochastic knapsack problem~\cite{marchetti1995stochastic}. 
The problem is challenging and it has been proved that there is no online algorithm that achieves any non-trivial competitive ratio~\cite{marchetti1995stochastic}. 
In Phoebe, we propose a simpler threshold-based algorithm that takes into account the arrival rate of jobs and the distributions for the weights and value-to-weight ratio based on the model estimation. The intuition behind this approach is that we want to select the items that are most ``cost-effective'' and, ideally, an optimal policy will select the items with high values of $\pi$. With higher job arrival rates, the probability of selecting each job is further reduced. 
Thus, if an item with estimated weight $w_i$ arrives with estimated value-to-weight ratio of $\pi_i$ at time $t$, when the remaining resource capacity $n(t)$ is larger or equal to its weight $w_i$ and its value-to-weight ratio is larger than the predefined threshold, $\pi^*$, the item is accepted, otherwise the item is rejected. That is, 
{\small
\begin{align}
D\big(w_i,\pi_i,n(t)\big)=
\begin{cases}
1~~\text{if}~~\pi_i\geq\pi^*~~\text{and}~~n(t)\geq w_i,\\ 0~~\text{otherwise,} 
\end{cases}
\end{align}}
where, $D\big(w_i,\pi_i,n(t)\big)$ is the binary decision variable to accept/reject an item $i$.
This policy ensures that in the set of accepted items, their average estimated value-to-weight \rev{ratios are} larger than $\pi^*$. 
Assuming that the two distributions of weight and value are independent, we can define:\vspace{-2ex}
{\small
\begin{align}
\pi^*&=\Phi_\pi(1-p), \text{and}\\
p &= \frac{W}{\lambda T \mathbb{E}_w(w)}, \label{eq:p}
\end{align}}
where $\Phi_\pi(\pi)$ is the cumulative distribution function for $\pi$, and $p$ denotes the ratio between the total resource capacity, i.e., the total global available storage space, and the expected total weights for all the arriving items. The denominator calculates the expected total weights by taking the product of the arrival rate $\lambda$, $T$, according to Little's Law~\cite{little1961proof}, and the expected weight for the items, $\mathbb{E}_w(w)$, assuming that the weight distributions for all items are i.i.d. Without capacity constraints, one can show that this threshold of $\pi^*$ results in the expected total weights for the selected items to be equal to the total capacity.

\subsection{Discussion}\label{sec:discussion}

We now discuss the trade-offs with our checkpoint optimizer.
The heuristic method for selecting the checkpoint stages is simple and fast.
It does not require solvers for the optimization formulation, however it involves sorting stages as a pre-processing step and uses brute-force search for identifying the optimal set of checkpoints. Since we enumerate over all feasible solutions, considering the global storage constraint at the same time would be expensive as for each possible cut, the global data size needs to be computed by examining the full execution graph. Therefore, we incorporate the global storage constraint in a separate step, which dramatically reduces the computation complexity.
The heuristic method, however, is not as flexible as the holistic integer programming solution in terms of adding additional constraints or considering multiple cuts.
However, practically, it is more desirable to create single checkpoints in more jobs than multiple checkpoints in a given job.

{\color{black}The same formulation as in \verb|OptCheck| and \verb|OptCheck2| can be generalized to optimize the checkpoint decisions jointly across multiple jobs, e.g. jobs in the queue. However, the computation complexity increases significantly. In Phoebe, we consider a two-step approach to (1) choose the most profitable jobs for the checkpoint using an online stochastic knapsack framework as  shown in Section~\ref{sec:capa}; and (2) pick the optimal checkpoints for the selected jobs.}

Finally, \rev{Phoebe's} architecture is flexible: modules can be replaced or removed depending on the performance requirement.
For instance, the ML predictor for TTL can be removed as the optimizer can use the estimated TTL from the job runtime simulator directly. In the extreme case, the optimizer can use some basic ``prediction'' for the stage costs, e.g., assuming a constant output size and runtime for each stage and simply make decisions based on the query execution graph and the outputs of the job runtime simulator (results can be seen in Section~\ref{sec:result}). In this sense, the cuts are mostly selected based on the count of stages. %

%% file: Results.tex
\section{Experiments}\label{sec:result}

In this section, we evaluate Phoebe (1) using back-testing over production Cosmos workloads; and (2) executing a smaller set of jobs for impact on performance.
We first evaluate the accuracy of stage-wise cost models (output size and execution time), and then evaluate the effectiveness of checkpoint optimization for two applications, namely, freeing up local storage on hotspots and quickly restarting failed jobs. We also provide anecdotal evidence of how Phoebe can help other checkpoint applications.
For all experiments, we used the production workload for \scope jobs over different days (with hundreds of thousands of jobs per day).

\subsection{Stage-Level Predictors}\label{sec:result_cost}
We evaluated the prediction accuracy for both the LightGBM and DNN models as discussed in Section~\ref{sec:module1}.

\begin{figure}[t]
	\centering
	\begin{minipage}[b]{0.16\textwidth}
		\includegraphics[width=\textwidth]{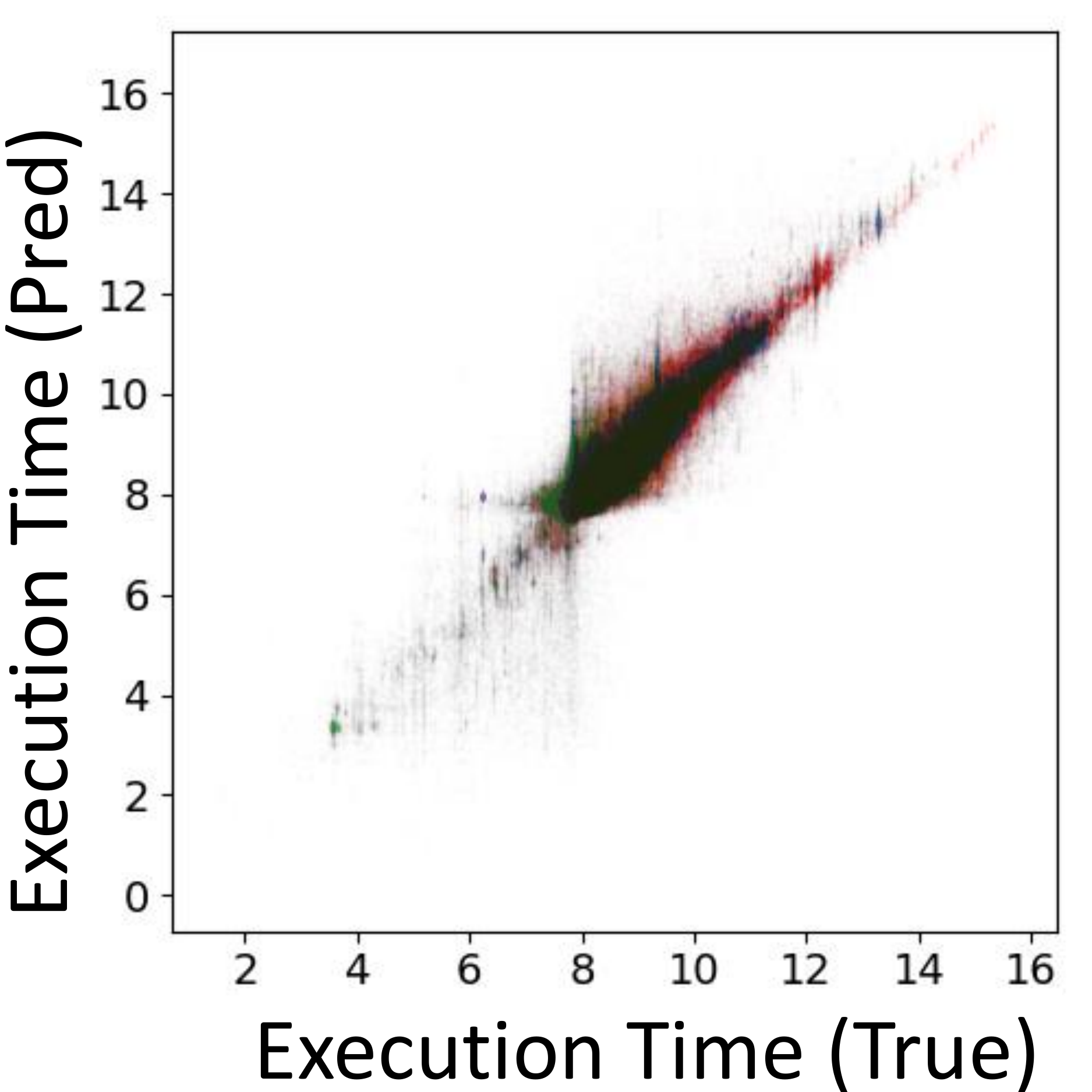}
	\end{minipage}%
	\begin{minipage}[b]{0.16\textwidth}
		\includegraphics[width=\textwidth]{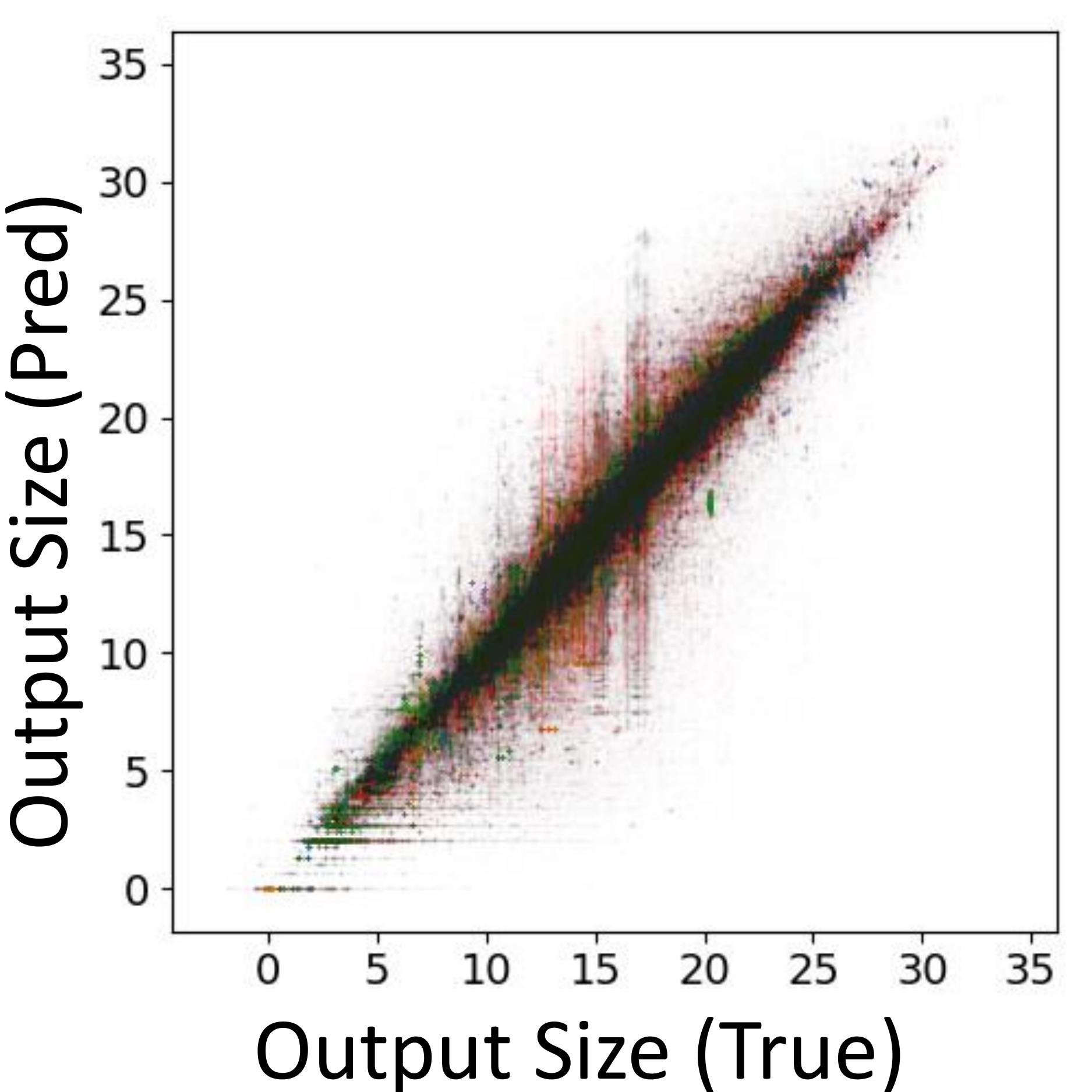}
	\end{minipage}%
	\begin{minipage}[b]{0.16\textwidth}
		\includegraphics[width=\textwidth]{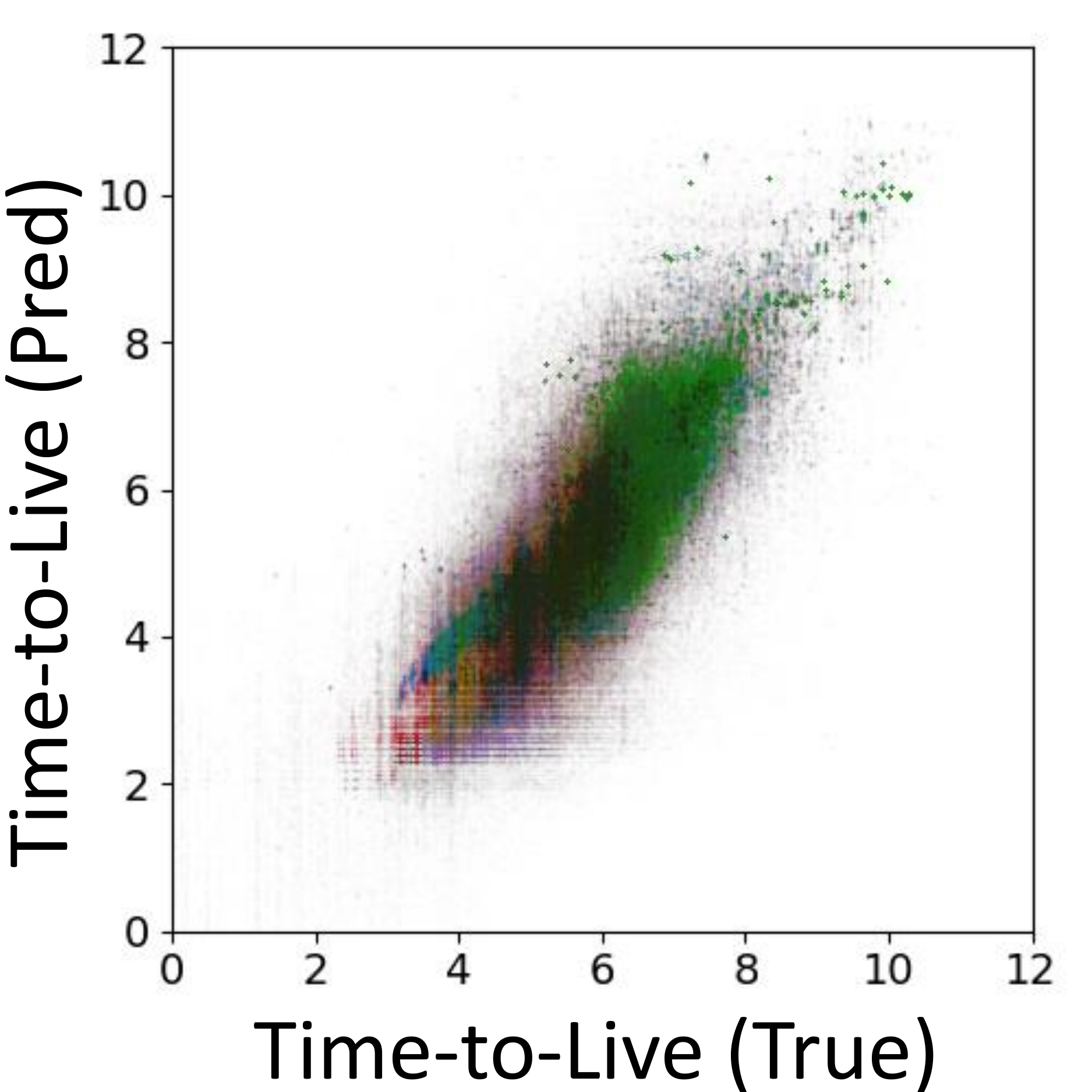}
	\end{minipage}%
	\vspace{-0.15cm}
	\caption{Accuracy for LightGBM models for predicting stage execution time (left), stage output size (middle), and time-to-live (right).}\label{fig:nimbus}
	\vspace{-0.1cm}
\end{figure}

\subsubsection*{LightGBM Models}
We developed the general and stage-type specific LightGBM models {\color{black}with 5-days data (with 13.0 million samples) and testing on 1-day data (with 2.9 million samples)}. 
Figure~\ref{fig:nimbus} shows the prediction accuracy for the stage-type specific models color-coded by the stage types. 
Since there are 33 stage-types, we have 33 models each for predicting the execution time and the output size. The $R^2$ values for the LightGBM models are 0.85 and 0.91 for the execution time and output size predictions respectively. 

The $R^2$ for predicting the time-to-live (TTL) is 0.35 (see right of Figure~\ref{fig:nimbus}), which is not as good due to slightly over-estimation. Future work can focus on improving the accuracy by incorporating more specific rules in the simulator or applying more fine-grained stacking models, such as dependency-type specific models. However, the correlation between the prediction and the true value is relatively high (0.77). Note that, in the optimization module, the absolute values for TTL are not as important as the relative scale. A model that can capture the correct order of the TTL is anyways helpful in improving the checkpoint objective (see Section~\ref{sec:saving}).

Based on Permutation Feature Importance (PFI)~\cite{breiman2001random}, the top 5 important features for one of the trained cost models are: \textit{Estimated Exclusive Cost} (0.75), \textit{Estimated Cardinality} (0.13), Historic \verb|MergeJoin| Latency (0.10), \textit{Estimated Input Cardinality} (0.06), and Historic \verb|Reduce| Latency (0.06), measured by the reduction of $R^2$ if shuffling the corresponding entry of the features. We can see that \textbf{it is the mixture of estimated cardinality from the query optimizer and the historic information that jointly improves the prediction accuracy.}

To better understand the implication of the models, we ran additional experiments to use the {\it perfect cardinality estimation} as inputs. The $R^2$ metric is improved only by 0.04-0.05, which indicates the effectiveness of our approach that automatically adjusts for the biases in the inputs (i.e., cardinality estimates). \rev{If we use stage-type as features, the accuracy is not as good: for the prediction of the output sizes, the $R^2$ is reduced from 0.91 to 0.84; for the execution time instead the $R^2$ is reduced from 0.85 to 0.72.}

Another important measurement of performance is the generalization of the trained model to future days, which determines the frequency needed to retrain the ML models. In production, we need to determine the frequency for retraining and deploying new models to keep up with the changes in data distribution. In Figure~\ref{fig:decay}, we measure the performance of the trained model on testing data that is further away from the time frame where the training data is extracted from (e.g., 1 day after, 2 days after, etc.). We can see that the accuracy reduces gradually as the testing dates move over time. 

\begin{figure}[t]
	\centering
\scalebox{0.85}{
\begin{tikzpicture}  
 
\begin{axis}  
[  
ybar,
width=8.5cm,
height=4.5cm, %
enlargelimits=0.15,%
legend style={at={(0.5,1.25)}, %
	anchor=north,legend columns=-1},     
ylabel={R Square}, %
symbolic x coords={1-Day, 2-Day, 7-Day, 14-Day},  
xtick=data,  
nodes near coords,  
nodes near coords align={vertical},  
]  
\addplot coordinates {(1-Day,0.82) (2-Day, 0.80) (7-Day, 0.73) (14-Day, 0.64)}; %
\addplot coordinates {(1-Day, 0.88) (2-Day, 0.88) (7-Day, 0.86) (14-Day, 0.80)};  
\legend{Execution Time, Output Size}  

\end{axis}   
\end{tikzpicture}}  
	\vspace{-0.33cm}
 	\caption{Performance accuracy with testing days further away from the training period.}\label{fig:decay}
	\vspace{-0.1cm}
\end{figure}
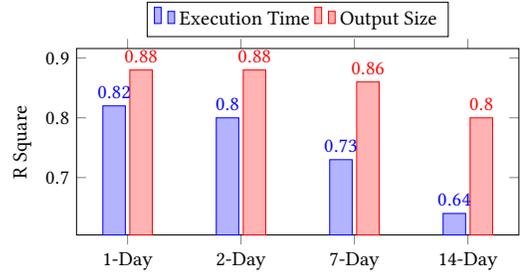

\subsubsection*{DNN Models}
We developed the general DNN models for both the execution time and output size using the same data. 
The $R^2$ values for the DNN models are 0.84 and 0.89 for the execution time, and output size prediction, respectively.  

{\color{black}\subsubsection*{Benchmark of Existing Models}

\begin{figure}[t]
	\centering
	\includegraphics[width=\columnwidth]{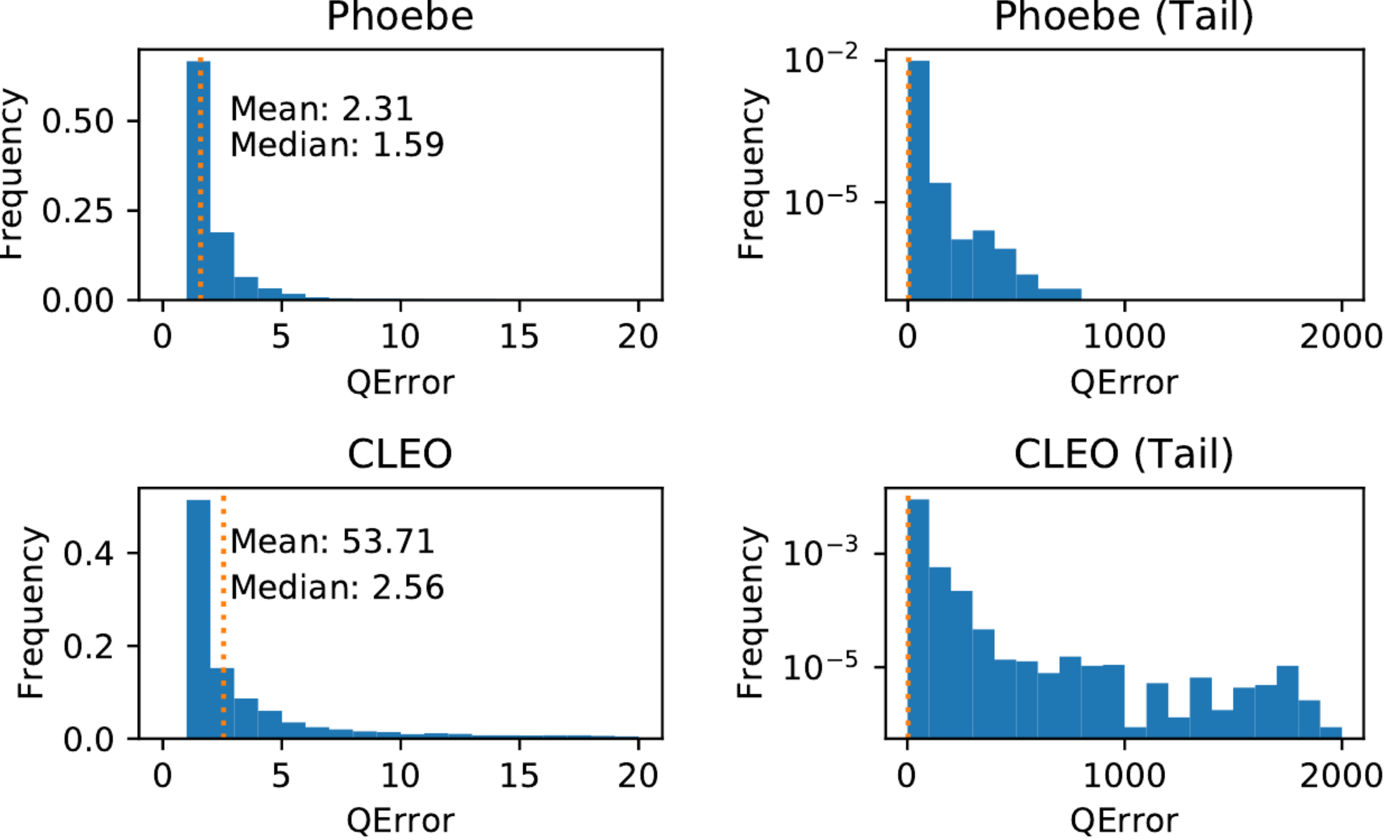}
	\vspace{-4.5ex}
	\caption{{\color{black}Accuracy for end-to-end job run-time prediction.}}\label{fig:TR}
	\vspace{-0.2cm}
\end{figure}
  
Figure~\ref{fig:TR} shows the distribution of prediction accuracy for the end-to-end job execution time from Phoebe (top) compared with CLEO~\cite{siddiqui2020cost} (bottom), measured by QError~\cite{moerkotte2009preventing} in origin scale (e.g., in seconds). The QError is a commonly used metric for measuring the accuracy of cardinality estimates~\cite{negi2020cost}. It is defined by:
	$\text{QError}(y,\hat{y}) = \max \{y/\hat{y}, \hat{y}/y\}$ where $y$ and $\hat{y}$ denote the actual and predicted values respectively.
The QError of CLEO has a long tail (see right of Figure~\ref{fig:TR}), meaning that for a small portion of jobs, the estimation is not good, and those are usually long-running jobs (>66\% longer on average than all the jobs). This is consistent with our observation that the cost models have lower accuracy for large complex query plans. 
} %

\subsubsection*{Discussion}
\textbf{The training time for DNN models is much longer than the LightGBM models.} Each (general) DNN model takes over 40 hours with a standard virtual machine with 6 Cores, 56 GB RAM, 
and \rev{an} NVIDIA Tesla K80 GPU. This is partially due to the large size of the training data as well as the introduction of the LSTM layer for the featurization of the text columns.
One potential improvement is to replace it with attention layers such as transformers to allow better parallelization with GPU.
Compared with the general DNN models, the stage-type specific model based on LightGBM is slightly more accurate. And the training time is much shorter, in the order of minutes. Therefore, in the following sections, we use the prediction results from LightGBM as the input to the optimizer.

\subsection{Freeing Temp Data Storage in Hotspots}\label{sec:saving}

We now evaluate the checkpoint algorithm for freeing up temp data storage on local SSDs in hotspots using back-testing, i.e., to see how well the algorithm would have done ex-post. 

\subsubsection*{Integer Programming}
We implemented the integer programming (formulation for freeing up temp data storage subject to different numbers of cuts as well as the cost for global data storage). 
We used Google OR-Tools~\cite{ortools} with CBC solver~\cite{forrest2018jpfasano} on Azure ML~\cite{aml}. 
Indeed, the IP solution is more general since it considers multiple cuts and different costs for the global storage. 
However, as illustrated in Figure~\ref{fig:time}, the processing time of using solvers is two orders of magnitude longer than the heuristic solution that looks for the optimal checkpoint time. The IP solutions provide interesting insights:

\begin{figure}[t]
	\centering
	\scalebox{0.9}{
	\begin{tikzpicture}  
	
	\begin{axis}  
	[  
	ybar,
	width=8.5cm,
	height=4.5cm, %
	enlargelimits=0.15,%
	legend style={at={(0.5,1.25)}, %
		anchor=north,legend columns=-1},     
	ylabel={Runtime (Second/Job)}, %
	symbolic x coords={Heuristic, IP with 1 Cut, 2 Cuts, 3 Cuts},  
	xtick=data,  
	nodes near coords,  
	nodes near coords align={vertical},  
	]  
	\addplot coordinates {(Heuristic,1.09) (IP with 1 Cut, 132) (2 Cuts, 180) (3 Cuts, 222)}; %
	\end{axis}   
	\end{tikzpicture}}  
	\vspace{-0.2cm}
	\caption{Runtime comparison of the heuristic solution with IP solution having different \rev{numbers} of cuts.}\label{fig:time}
	\vspace{-0.3cm}
\end{figure}
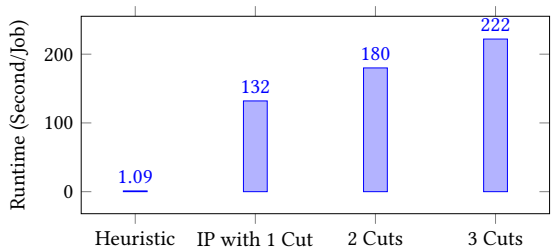

\begin{itemize}
	\item \textbf{Adding more cuts is not cost-effective.} Figure~\ref{fig:pareto} shows the Pareto frontier for freeing up temp data storage. The x-axis shows the median usage for the global storage normalized by the total original temp data usage across all jobs. The y-axis shows the median temp data saving also normalized by the original temp data usage. {\color{black}As we can see, adding more cuts is only helpful for large jobs (with >14 GB*Hour temp data usage), and not for all jobs. Therefore, we can have different checkpoint strategies for different types of jobs. %
	} 
	\item \textbf{There are jobs with ``free'' cuts.}
	Based on the IP formulation that inserts a high cost for global storage, we found 
	jobs with sub-graphs that are independent of each other, resulting in {\it free} checkpoints, i.e., not requiring any global storage.
\end{itemize}

\begin{figure}[t]
	\centering
	\includegraphics[width=0.87\columnwidth]{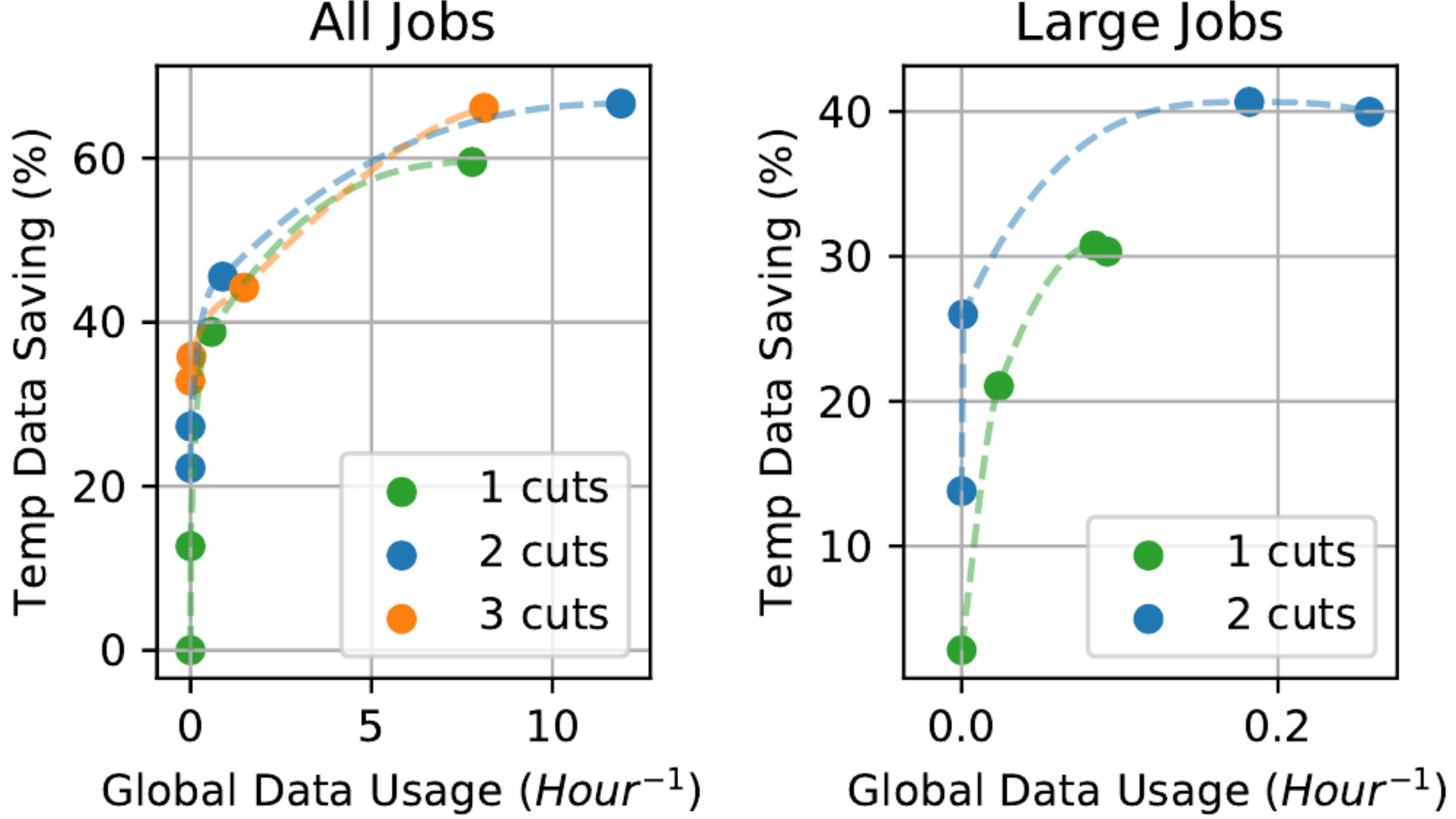}
	\vspace{-1.5ex}
	\caption{{\color{black}Pareto frontier for multiple cuts.}%
	}\label{fig:pareto}
\end{figure}

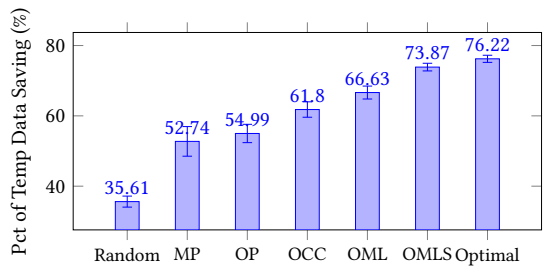
\begin{figure}[t]
	\centering
	\scalebox{0.9}{
	\begin{tikzpicture}  
	
	\begin{axis}  
	[  
	ybar,
	width=8.5cm,
	height=4.5cm, %
	enlargelimits=0.15,%
	legend style={at={(0.5,1.25)}, %
		anchor=north,legend columns=-1},     
	ylabel={Pct of Temp Data Saving (\%)},
	x tick label style={font=\small,text width=1.5cm,align=center},
	symbolic x coords={
		Random,
		MP, 
		OP, 
		OCC, 
		OML,
		OMLS,
		Optimal
	},  
	xtick=data,  
	nodes near coords,  
	nodes near coords align={vertical},  
	]  
	\addplot+[ error bars/.cd,
	y dir=both,
	y explicit relative]
	 coordinates {
	(Random, 35.612171) +- (0,0.043976)
	(MP, 52.74)+- (0,0.079886294)
	(OP, 54.991667)+- (0,0.047391)
	(OCC,  61.798333)+- (0,0.035442)
	(OML, 66.632484)+- (0,0.027596)
	(OMLS, 73.869763)+- (0,0.014726)
	(Optimal,  76.223741)+- (0,0.013294)
	}; %
	\end{axis}   
	\end{tikzpicture} } 
	\vspace{-0.3cm}
	\caption{{\color{black}Percentage of temp data storage saving for different checkpointing approaches.}}\label{fig:saving}
	\vspace{-0.3cm}
\end{figure}

\subsubsection*{Optimal Checkpoint Time}
Figure~\ref{fig:saving} shows the daily average percentage of temp data storage saving for the different approaches based on 6-day's observations with approximately 405,000 jobs in one of the \cosmos clusters. The error bar shows the standard deviation. We consider the following approaches: 
\begin{enumerate}
	\item \textit{Random}: using a random checkpoint selector that randomly selects stages as the global checkpoints; 
	\item {\color{black}\textit{Mid-Point (MP)}: using the estimated stage scheduling from the job runtime simulator, cut the execution graph into two based on the \textit{mid-point} of the total job execution time; 
	} 
	\item \textit{Optimizer + Estimated Cost (OP)}: using the estimated cost from the query optimizer for the output size/execution time for each stage as the inputs to the job runtime simulator and the proposed optimizer; 
	\item \textit{Optimizer + Constant Cost (OCC)}: similar to \textit{Optimizer + Estimated Cost}, assuming a simple constant for the output size and execution time for each stage;
	\item \textit{Optimizer + ML Cost Models (OML)}: the proposed optimizer that uses ML predictions as inputs and uses the estimated TTL from the job runtime simulator;
	\item \textit{Optimizer + ML Cost Models + Stacking Model (OMLS)}: the proposed optimizer using the TTL from the stacking model;
	\item \textit{Optimal}: the optimal off-line checkpoint optimizer based on the knowledge of the true output size and TTL for all stages. 
	
\end{enumerate}

\begin{figure}[t]
	\centering
	\includegraphics[width=\columnwidth]{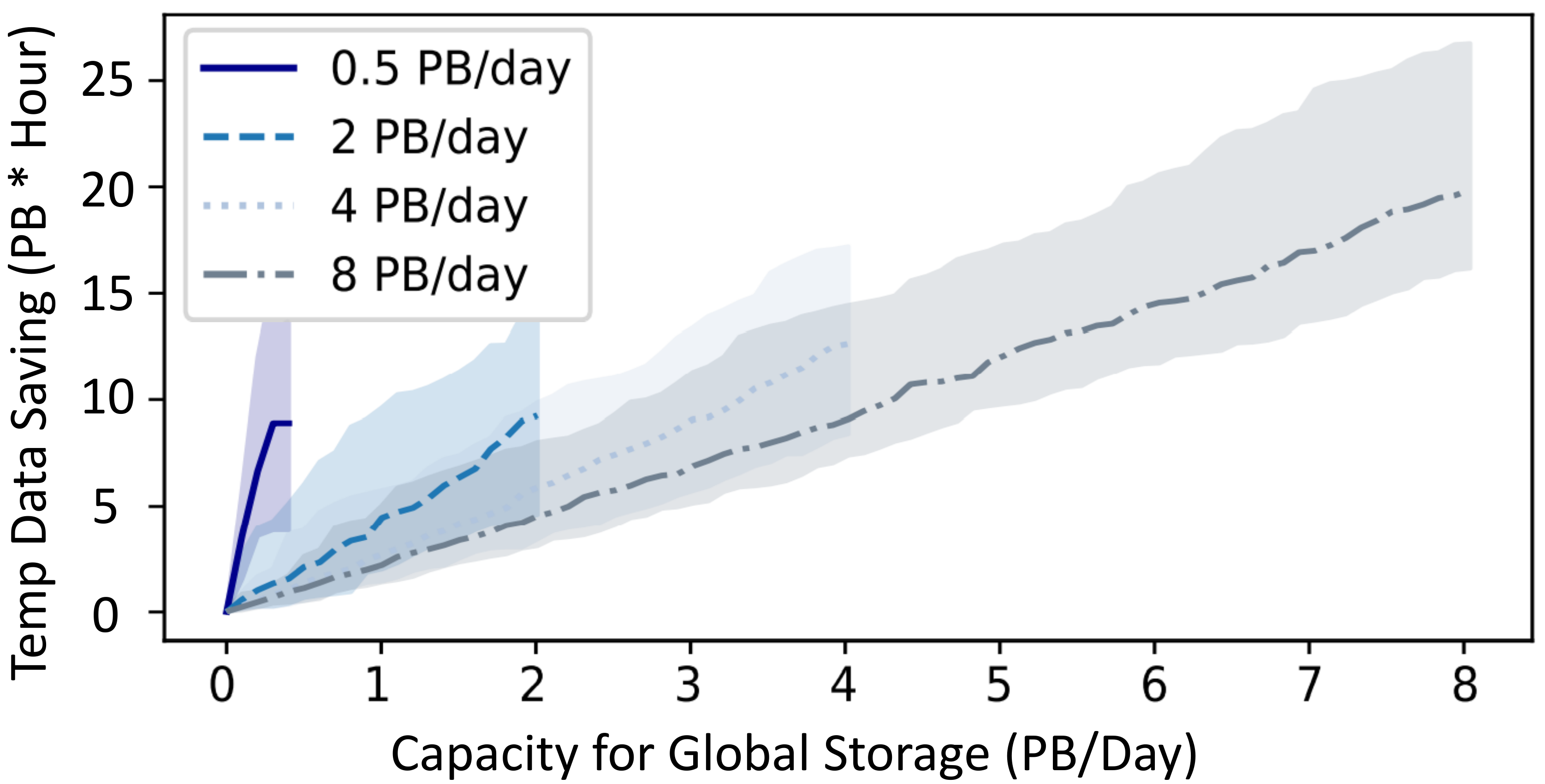}
	\vspace{-0.5cm}
	\caption{Cumulative temp data storage saving as a function of global storage used.}\label{fig:saving2}
	\vspace{-0.3cm}
\end{figure}

\begin{figure*}[t]
	\centering
    \includegraphics[width=0.83\textwidth]{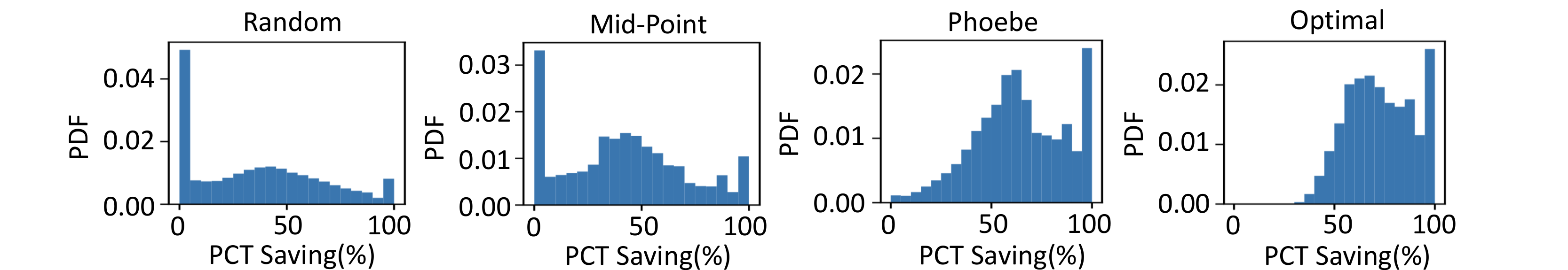}
	\vspace{-0.2cm}
	\caption{Distribution of expected recovery time saving (at the job level) for failed jobs.}\label{fig:rec79}
	\vspace{-0.1cm}
\end{figure*}

We can see that the random optimizer can only free up 36\% of temp data storage on average, measured by the portion of temp data storage that can be cleared in the unit of Petabytes*Hour (PB*Hour). 
By further improving on the estimation of the output size, our proposed optimizer without the stacking model can free up to 67\% of temp data storage on average, while with the stacking model the percentage further increases to 74\%. This is not very far from the optimal off-line optimizer, which frees up to 76\%. 
While we thought that using the cost from the optimizer (Estimated Costs) can improve the performance compared to using the constant cost, the corresponding temp data saving is actually  worse, which is due to the large errors in cost estimation from the optimizer. 
Thus, \textbf{using the job runtime simulator and learned stage-wise costs significantly increases the saving in temp data storage on hotspots.}

Figure~\ref{fig:saving2} shows the temp data saving with respect to the cumulative usage of the global data storage under different capacity constraints for the global storage. The $5^{\text{th}}$ and $95^{\text{th}}$ confidence level on the saving is also shown. We can see that with a larger capacity, the total temp data saving increases. However, because we are less selective about the jobs, the average value-to-weight ratio for the selected pool of jobs and the slope for the curve decreases. Therefore, to determine the best size of the capacity limit for the global storage, we should take into account the corresponding costs for the two types of storage and determine the optimal value.

\subsection{Restarting Failed Jobs}
We now present the experiment results for the heuristic method for maximizing the recovery time when restarting failed jobs.
We used 1-day's data with approximately 62,000 jobs, and we evaluate 4 algorithms to select the optimal checkpoints:
\begin{enumerate}
	\item \textit{Random}: using a random checkpoint selector that randomly selects stages as the global checkpoints; 
	\item \textit{Mid-Point}\footnote{{\color{black}Note that the cost-based approach proposed by \cite{salama2015cost} for minimizing the potential maximal total cost (execution cost and materialization cost for the dominant execution path) through enumeration is  infeasible in \cosmos because of the size of the DAGs and the presence of spool operators~\cite{leeka2019incorporating}. Mid-point is a simple heuristic that extracts the ``longest'' execution path according to the stage execution schedule, and picks the mid-point timestamp and the corresponding last finished stages as the checkpoints. This strategy maximizes the total cost reduction on the dominant path.}}: \rev{cut the execution graph based on the mid-point of the total job execution time based on estimated scheduling;}
	\item \textit{Phoebe}: the proposed method that uses the model prediction for the stage scheduling as well as the estimated probability of failures for each stage.
	\item \textit{Optimal}: the optimal off-line checkpoint optimizer based on the knowledge of the actual stage scheduling and failure probabilities for all stages. 
\end{enumerate}

Figure~\ref{fig:rec79} shows the distribution of the percentage of recovery time saving for the 4 algorithms at the job level. 
The average percentage saving for the expected recovery time is 36\% for the Random algorithm, 41\% for the Mid-Point algorithm, 64\% for Phoebe, and 73\% for the Optimal algorithm. We can see that by introducing the predictions on stage scheduling and the estimated failure probabilities, the recovery time saving improves. However, more work can be done to further improve the estimation accuracy for the TFS.

\subsection{Overheads and Production Test}
Phoebe adds the following overheads on top of the \scope compiler: (1)~metadata and model lookup:~$15$ms on average; (2)~scoring and optimization:~$1.09$s; and (3)~query optimization, which is negligible. As we materialize checkpoints by adding an additional stage that executes in parallel, the overhead for data writing is usually hidden by other parts of the job execution plan. 
In sum, we are expecting approximately ~1s overhead in total compared to several minutes of end-to-end job compilation, which is acceptable.

We deployed Phoebe in the production environment and applied the checkpoint mechanism to over 1000 random analytic jobs (514 hours of total job execution time). The median increase in latency was just 1.8\%. {\color{black}If a more constrained overhead is required, we can select simpler predictors as for example suggested in Figure~\ref{fig:saving}.} We tested on another 256 large jobs (with >1h job runtime). Figure~\ref{fig:perf} shows the distributions of percentage impact on latency and IO time. While the IO time for some jobs increased by >20\%, the median increase for latency is only 2.6\%. Among those long-running jobs, the average percentages of data checkpointed and temp data storage saved are 12.3\% and 48.6\% respectively. 

\begin{figure}[t]
	\centering
	\vspace{0.1cm}
			\includegraphics[width=\columnwidth]{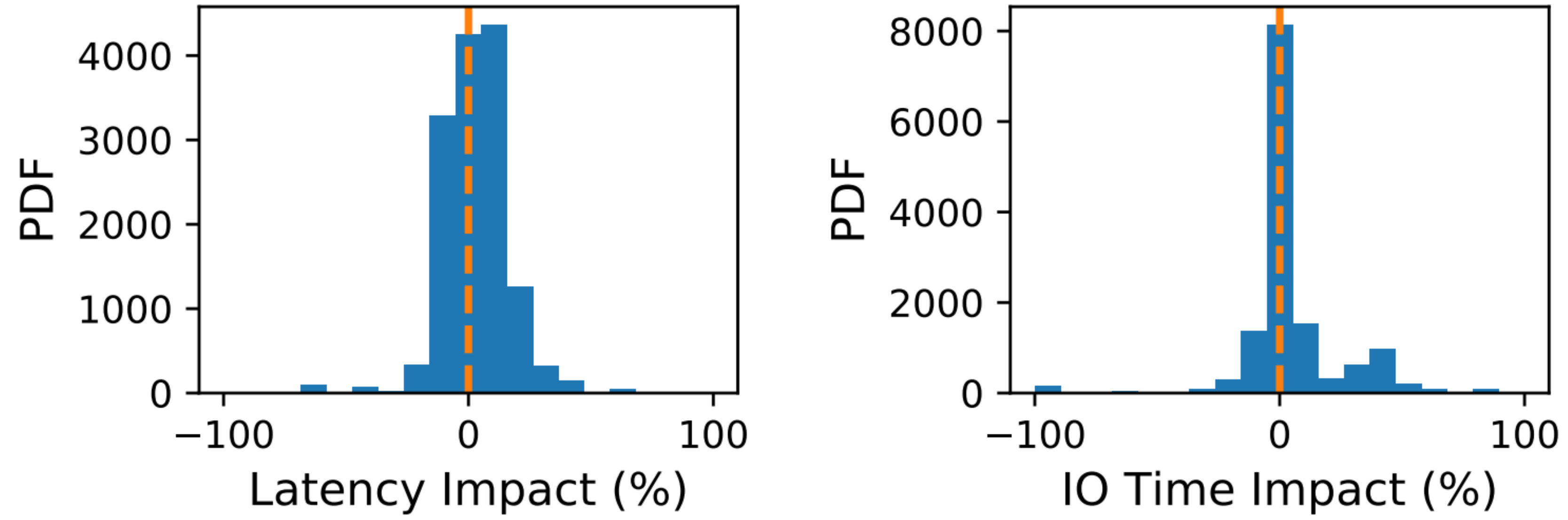}
	\vspace{-0.4cm}
	\caption{Checkpointing impact on latency and IO time.}\label{fig:perf}
	\vspace{-0.6cm}
\end{figure}

\subsection{Other Checkpoint Applications}

We present a couple of anecdotal evidence on how Phoebe can also help other checkpoint applications. 
For instance, in one of our new SKUs in \cosmos, the local SSDs are not scaled in the same ratio as CPU cores. As a result, accommodating similar volumes of temporary storage would preclude from leverage the full compute capacity. Using Phoebe, we could reduce the temp storage load, and help increase the number of containers per machine by up to $28\%$.
In another application, extremely large \scope jobs could run for several hours, making the query plan highly sub-optimal. 
Splitting one such large production job into smaller ones {\color{black} that have more accurate cost estimation thus better-optimized query plans }resulted in its runtime to decrease from $30+$ to $20+$ hours.

%% file: Related.tex
\section{Related Work}
\label{sec:related}

In this section, we describe the different checkpointing mechanisms and systems. 
We also discuss the recent machine learning based approaches that predict characteristics like cardinality, operator cost, and resource allocation for query plans.

\stitle{Runtime Checkpointing.} This set of techniques make checkpointing decisions while the job is running. They are dynamic and more suitable for black-box workloads.
By optimizing the checkpointing intervals~\cite{daly2006higher}, runtime checkpointing is easier to optimize and more applicable for transient resource environments. 
For example, Flint~\cite{sharma2016flint} and TR-Spark~\cite{yan2016tr} focused on optimizing the checkpoint intervals and the selection of servers/tasks for replication to minimize the application runtime.

\stitle{Compile-time Checkpointing.} This set of techniques leverage query characteristics and propose optimal checkpoints at the task or operator-level. 
For example, FTOpt~\cite{upadhyaya2011latency} developed a cost-based fault-tolerance optimizer to determine the (1) fault-tolerance strategy and (2) the frequency of checkpoints for each operator in a query plan. 
{\color{black}FTOpt works with non-blocking query plans only.
Osprey~\cite{yang2010osprey} proposed the intra-query checkpoint mechanism for a distributed database system, and preserved the intermediate outputs for all sub-queries.
But Osprey does not support general workloads with self-joins and nested queries.}
Similarly, ~\cite{chen2013selective} proposed a divide-and-conquer algorithm to solve a similar optimization problem. %
However, they do not consider temporary data saving costs or enforce any global storage constraints. 
The authors in~\cite{salama2015cost} describe a cost-based checkpoint approach 
by enumerating all possible query plans and materialization configurations.
Their algorithm relies on the cardinality estimates provided by the optimizer to find the query plans with the shortest dominant path under mid-query failures.
However, prior work~\cite{wu2018towards} has shown that big data query optimizers often under / over estimate cardinalities. 
They are also not applicable for distributed systems as they do not consider estimating costs for stages instead of operators. %
Production systems frequently encounter large jobs with hundreds or even thousands of stages, where the above algorithms prove to be inefficient.

\stitle{Checkpointing in streaming systems.} The checkpointing problem has also been studied in the context of stream processing systems.
In streaming systems, %
fault-tolerant operator implementations employ logging~\cite{balazinska2005fault}, replication~\cite{kwon2008fault}, and scaling out techniques~\cite{castro2013integrating} to reduce the cost of failure recovery.
In contrast, we employ a restart-after-failure mechanism that minimizes the wasted computation in large-scale analytical programs.

\vspace{0.2cm}
We now review prior work that uses machine learning for cost prediction and query optimization.

\noindent{\bf Cost Models and Query Optimization}
~\cite{marcus2019plan} proposed a \textit{plan-structured deep neural networks} to incorporate information from individual operator level to the full query plan and predict the job performance. 
Neo~\cite{marcus2019neo} used a similar DAG representation, the tree convolution~\cite{mou2016convolutional}, to predict the total execution latency for a given query plan. By searching over the space of the plans, the model discovered the optimal plan with the minimum expected execution time.
\cite{ortiz2018learning} introduced the \textit{subquery representation} for each state. The concept is similar to the \textit{internal neural unit}~\cite{marcus2019plan}, i.e., to concatenate the featurization of a sub-query
with a new operation for the representation of a larger sub-query. Unfortunately, DNN solutions are not suitable for larger complex query graphs, that are found in big data systems such as Cosmos, due to the fact that many CNN layers will be stacked, {\color{black}which results in gradient vanishing or explosion. This is the same problem faced by DNN models before residual/highway connections were introduced~\cite{resnet,hsu2016exploiting}. Whether a similar approach applied in our case is future work.~\cite{kaoudi2020ml} proposed to featurize query plans using topological features which doesn't capture the detailed dependency between operator and stages, therefore is too general to capture the heterogeneity among complex query plans.
Similar to CLEO~\cite{siddiqui2020cost},~\cite{akdere2012learning} used prediction of child operators as input to ML models to predict the parent operators' cost. This model also leads to error prorogation that is not suitable for large query plans.}
ML has also been used to improve the estimation for cardinality~\cite{wu2018towards,dutt2019selectivity,MCSN}.
Other works have considered persisting intermediate results that overlap across queries~\cite{singh2016progressive,jindal2018selecting} for computation reuse.
In contrast, we focus on stage-wise cost models.
To the best of our knowledge, Phoebe is the first system to propose cost models at the stage-level in a large distributed system with large, complex production workloads.

\eat{
\stitle{Capacity Constraints and the Knapsack Problem.}
\cite{zhou2008budget} proposed an \textit{Online-KP-Threshold} algorithm that is upper and lower bounded for the competitive ratio based on the range of value-to-weight ratio.
However, in our case, as the distribution for the value-to-weight ratio has a long tail (up to multiple scales), the bound is not satisfactory. 
~\cite{bello2016neural} proposed to use reinforcement learning to solve the combinatorial optimization problem such as the Travel Salesman Problem (TSP) and the knapsack problem. 
Based on known distribution(s) for the value and weight, \cite{kleywegt1998dynamic} and \cite{kleywegt2001dynamic} introduced Markov Dynamic Programming (MDP) to solve the online stochastic knapsack problem, and the optimal policy of selecting the items is based on a threshold as a function of time and item value/weight. However, the policy requires solving a partial differential equation and known distributions for the value/weight. In our use cases, as the value and weight is unknown at the compile time, the decision is made based on the estimation. Therefore, the assumptions in~\cite{kleywegt2001dynamic} no longer hold. 
}

\vspace{-0.2cm}

%% file: Conclusion.tex
\section{Conclusion}\label{sec:conclude}

This paper revisits the checkpointing problem in the context of big data workloads, with complex query DAGs.
We introduce Phoebe, a learning-based checkpoint optimizer to 
select the optimal set of stages to checkpoint, subject to different objective functions and storage constraints.
Phoebe leverages multiple machine learning models built upon state-of-the-art cost models at the operator level, to accurately predict the cost of each stage, including the execution time, the output size, and the time-to-live. Based on the estimated costs, Phoebe's checkpoint optimizer selects the optimal set of stages that maximize the checkpointing objectives, such as freeing up temp data on local SSDs, or recovery time for failed jobs. 
We validated Phoebe using production workloads, with results showing that Phoebe is able to free up >70\% of temp data storage on hotspots, improve the recovery time of failed jobs by >60\% on average, and yet having acceptably low performance impact (less than 3\%).